\documentclass[sigconf]{acmart}
\renewcommand\footnotetextcopyrightpermission[1]{} % removes footnote with conference information in first column
\usepackage{booktabs} % For formal tables
\allowdisplaybreaks
\usepackage{flushend}
\usepackage{bbm}
\usepackage{bm} % boldface math symbols
\usepackage{enumitem}
\usepackage{subcaption}
\usepackage{tikz-cd}

\newcommand{\thistheoremname}{}
\newtheorem*{genericthm*}{\thistheoremname}
\newenvironment{namedthm*}[1]
{\renewcommand{\thistheoremname}{#1}%
\begin{genericthm*}}
{\end{genericthm*}}

\newtheorem*{claim*}{Claim}

\newtheorem*{remark*}{Remark}
\newtheorem*{example*}{Example}
\setcopyright{rightsretained}
\renewcommand{\Pr}{\mathbb{P}} % probability measure
\newcommand{\expect}{\mathbb{E}} % expectation

\newcommand{\cdf}{cdf}

\newcommand{\supn}{^{(n)}}
\newcommand{\tod}{\Rightarrow}

\begin{document}
\title{Delay Asymptotics and Bounds for Multi-Task Parallel Jobs}

\author{Weina Wang}
\affiliation{%
\institution{Coordinated Science Lab\\University of Illinois at Urbana-Champaign}
\city{Urbana} 
\state{IL}
\country{USA}
}
\email{weinaw@cs.cmu.edu}

\author{Mor Harchol-Balter}
\affiliation{%
\institution{Computer~Science~Department\\Carnegie Mellon University}
\city{Pittsburgh} 
\state{PA}
\country{USA}
}
\email{harchol@cs.cmu.edu}

\author{Haotian Jiang}
\affiliation{%
\institution{Department of Physics\\Tsinghua University}
\city{Beijing}
\country{China}
}
\email{jhtdavid@cs.washington.edu}

\author{Alan Scheller-Wolf}
\affiliation{%
\institution{Tepper School of Business\\Carnegie Mellon University}
\city{Pittsburgh} 
\state{PA}
\country{USA}
}
\email{awolf@andrew.cmu.edu}

\author{R.\ Srikant}
\affiliation{%
\institution{Coordinated Science Lab\\University of Illinois at Urbana-Champaign}
\city{Urbana} 
\state{IL}
\country{USA}
}
\email{rsrikant@illinois.edu}

\begin{abstract}
We study delay of jobs that consist of multiple parallel tasks, which is a critical performance metric in a wide range of applications such as data file retrieval in coded storage systems and parallel computing.  In this problem, each \emph{job} is completed only when \emph{all} of its tasks are completed, so the delay of a job is the maximum of the delays of its tasks. Despite the wide attention this problem has received, tight analysis is still largely unknown since analyzing job delay requires characterizing the complicated correlation among task delays, which is hard to do.

We first consider an asymptotic regime where the number of servers, $n$, goes to infinity, and the number of tasks in a job, $k\supn$, is allowed to increase with~$n$.  We establish the asymptotic independence of any $k\supn$ queues under the condition $k\supn = o(n^{1/4})$. This greatly generalizes the asymptotic-independence type of results in the literature where asymptotic independence is shown only for a fixed constant number of queues.  As a consequence of our independence result, the job delay converges to the maximum of independent task delays.

We next consider the non-asymptotic regime.  Here we prove that independence yields a stochastic upper bound on job delay for any $n$ and any $k\supn$ with $k\supn\le n$.  The key component of our proof is a new technique we develop, called ``Poisson oversampling''.  Our approach converts the job delay problem into a corresponding balls-and-bins problem.  However, in contrast with typical balls-and-bins problems where there is a negative correlation among bins, we prove that our variant exhibits positive correlation.
\end{abstract}

\maketitle

\section{Introduction}

\subsection*{\normalsize{The problem}}
We consider a system with $n$ servers, each with its own queue.  Jobs arrive over time according to a Poisson process, and each job consists of some number of tasks, $k$, where $k\le n$.  Upon arrival, each job chooses $k$ distinct servers uniformly at random and sends one task to each server. Each server serves the tasks in its queue in a First-In, First-Out (FIFO) manner.  A job is considered to be completed only when \emph{all} of its tasks are completed.  Our goal is to compute the \emph{distribution} of \emph{job delay}, namely the time from when a job arrives until the whole job completes. If a job's tasks experienced independent delays, then computing the distribution of job delay would be easy: take the maximum of the independent task delays. Unfortunately, the task delays are not independent in general.

\begin{sloppypar}
Our model is a generalization on the classic \emph{fork-join} model, which is identical to our model except that it assumes that $k=n$: every job is forked to all $n$ servers.  In contrast, in our model, the fork is \emph{limited} to $k$ servers with $k\le n$.  So we will refer to our model as the \emph{limited fork-join} model.  Obtaining tight analytical job delay characterizations for fork-join systems is known to be notoriously difficult: exact analysis of fork-join remains an open problem except for the two-server case \cite{FlaHah_84,Bac_85}.
\end{sloppypar}

\subsection*{\normalsize{Motivation}}
Delay of jobs, rather than delay of individual tasks, is a more critical performance metric in systems with parallelism, yet a fundamental understanding of job delay is still lacking.  One example application is data file retrieval in coded storage systems \cite{JosLiuSol_12,ShaLeeRam_13,LiRamSri_16,LeeShaHua_17,ShaBouBac_17}.  Here a job is the retrieval of a data file, which is stored as multiple data chunks.  The data chunks are in a coded form such that any $k$-sized subset of them is enough to reconstruct the file.  Coded file retrieval can be modeled via the so-called $(n,r,k)$ model \cite{ShaLeeRam_13} where a job can request $r$ data chunks with $r\ge k$ and the job is completed as long as $k$ of them are completed.  Existing analysis of the $(n,r,k)$ model is usually not tight except for the light load regime \cite{JosLiuSol_12,LiRamSri_16}. The special case where $r=d$ and $k=1$, called the Redundancy-d model, is also highly non-trivial and was solved just last year \cite{GarHarSch_17}. Job delay in general $(n,r,k)$ models remains wide open.  Within the coded file retrieval setting, our limited fork-join model can be viewed as the $(n,k,k)$ problem.

Another application is parallel computing systems such as the ``map'' phase of the popular MapReduce framework \cite{DeaGhe_04}, where a job is divided into tasks that can run in parallel. A few papers have been written to analytically approximate the delay of MapReduce jobs. Please see Section~\ref{sec:related-work} for more details of related work.

In the above applications, load-balancing policies (see, e.g., \cite{YinSriKan_15,XiaLanAgg_16,LiRamSri_16,LeeShaHua_17,ShaBouBac_17} are usually used for assigning tasks to servers. For scenarios where either low-overhead is desired or information accessibility is constrained (such as in a distributed setting), workload agnostic assignment policies \cite{XiaLanAgg_16,LeeShaHua_17,ShaBouBac_17} can be preferred.  Our limited fork-join model assumes a random task assignment policy, which is suitable for such application scenarios.

\subsection*{\normalsize{Our approach and what makes this problem hard}}
The root of the hardness of analyzing job delay in our model is the complicated correlation among queues, which leads to the correlation among the delays of a job's tasks.  If the task delays were independent, then the probability distribution of job delay would have a simple form.  In this paper, we are interested in developing conditions and quantifying in what sense the job delay can be approximated by the job delay under the independence assumption.

\textbf{Asymptotic Regime.}
We first study a regime where we prove that a job's tasks can be \emph{viewed} as being independent:  We focus on the asymptotic regime where the number of servers, $n$, goes to infinity.  Here we are specifically interested in developing conditions under which the delays of a job's tasks are \emph{asymptotically independent}, i.e., their joint distribution converges to the product distribution of their marginals.

Asymptotic independence of a number of queues in large systems is often called ``chaoticity'' and studied under the name ``propagation of chaos.''  In many papers \cite{VulMicGod_12,XieDonLu_15,GarHarSch_16,GarHarSch_17_2}, asymptotic independence is simply assumed to simplify analysis. In some load-balancing settings, asymptotic independence has been proven (e.g., \cite{BraLuPra_12,YinSriKan_15}).  One strong restriction of the existing proofs is that only a \emph{constant} number of queues are proven to be asymptotically independent.  In contrast, our goal is to establish asymptotic independence for any $k$ queues where $k$ may grow with $n$; we write $k$ as $k\supn$ to explicitly indicate its dependence on $n$.  The asymptotic independence of any $k\supn$ queues implies the asymptotic independence of the delays of a job's tasks since they are sent to $k\supn$ queues.  Allowing $k\supn$ to grow with $n$ captures the trends that data files get larger and that jobs are processing larger and larger data sets \cite{CheAlsKat_12}.

When proving asymptotic independence of a constant number of queues in steady state, it is typical to start by showing asymptotic independence over a \emph{constant time interval} $[0,t]$, where $t$ is long enough for these queues to be close to steady state.  Unfortunately, since $k\supn$ grows with $n$ in our model, to reach steady state, the system needs a time interval $[0,\tau\supn]$, growing with $n$.  This further complicates the analysis since asymptotic independence then needs to be established over this longer, \emph{non-constant}, time interval.

\textbf{Non-asymptotic regime.}
Next, we study the non-asymptotic regime.  We show that for \emph{any} $n$ and \emph{any} $k\supn=k$ with $k\supn\le n$, the distribution of job delay is \emph{stochastically upper bounded} by the distribution given by independent task delays, which we call the \emph{independence upper bound}.  Therefore, independence not only characterizes the \emph{limiting} behavior of job delay in the asymptotic regime where $n\to\infty$, but also yields an upper bound for \emph{any}~$n$. I.e., the independence upper bound is asymptotically \emph{tight}. An illustration of the tightness is provided in Figure~\ref{fig:ccdf} generated from simulations.  The independence upper bound is also tighter than all the existing upper bounds in prior work \cite{RizPolCiu_16, LeeShaHua_17}.

We prove the independence upper bound using the theory of associated random variables \cite{EsaProWal_67}. Association (also called positive association) is a form of positive correlation, and it has the property that if a set of random variables are associated, then the maximum of them is stochastically upper bounded by the maximum of independent versions of them.  To show the independence upper bound, it thus suffices to show that the delays of a job's tasks are associated.  Such an association result is known for the classical fork-join model with $k\supn = n$, but not for the limited fork-join model when $k\supn<n$. When proving association, a commonly used idea is to observe the system at each job arrival time, and show that the numbers of tasks sent to different queues are associated \cite{NelTan_88,KumSho_93,ShaBouBac_17}. This corresponds to a balls-and-bins problem where $k\supn$ balls are thrown into $n$ bins in the same way that the tasks are sent to the queues. What is needed is that the numbers of balls thrown in different bins are associated, which is obviously true for $k\supn=n$ since they are all equal to one, but not true when $k\supn<n$.  In fact, they are actually \emph{negatively associated} by a classical result \cite{JoaPro_83}.  However, this does not mean that the \emph{steady-state} queues are negatively associated, leaving the association problem for $k\supn<n$ unsolved in the literature. As pointed out in \cite{LeeShaHua_17}, it was not known if independence yielded a bound, either lower or upper.

We develop a novel technique that we call ``Poisson oversampling,'' where we \emph{observe} the system not only when jobs arrive but also at the jump times of a Poisson process that is independent of everything else.  This oversampling does not change the dynamics in the system since it is only a way to observe the system state.  But now at each observation time, there could be one or zero job arrivals.  So in the corresponding balls-and-bins problem, there is certain probability that there are no balls at all.  By properly choosing the observation rate, this extra randomness surprisingly makes the numbers of balls thrown in any $k\supn$ bins (positively) associated, and further implies that the steady-state queues are associated.  With this technique, we are able to prove the independence upper bound for any $k\supn\le n$ for the first time.

\subsection*{\normalsize{Results}}
Our goal is to characterize the tail probability of the job delay distribution in steady state, since it is commonly used to quantify the quality of service.  We study a system with $n$ servers in which each job consists of $k\supn$ tasks.
 
Our first result is that under the condition $k\supn=o(n^{1/4})$, the queues at any $k\supn$ servers are \emph{asymptotically independent} in steady state as $n\to\infty$, and thus the delays of a job's tasks are also asymptotically independent.  It then follows that the job delay converges to the job delay given by the independence assumption.  This result is established in Theorem~\ref{thm:asym-independence} for generally distributed service times, and some explicit forms are given in Corollary~\ref{cor:exp} for exponentially distributed service times.  One might wonder where the order of $o(n^{1/4})$ comes from or whether it can be increased; we discuss this in Section~\ref{sec:explain}.  This is the first asymptotically tight characterization of job delay in the limited fork-join model.

Our next result is that for \emph{any} $n$ and \emph{any} $k\supn$ with $k\supn\le n$, the job delay is stochastically upper bounded by the job delay given by the independence assumption.  We refer to this upper bound as the \emph{independence upper bound}.  It is a new upper bound on job delay that is tighter than existing upper bounds.  The technique we develop for the proof, named ``Poisson oversampling'', may be of independent interest for other related problems.

\subsection*{\normalsize{Organization of the paper}}
The rest of this paper is organized as follows. Section~\ref{sec:related-work} discusses the related work. We introduce our model and notation in Section~\ref{sec:model-notation}. We summarize our main results in Section~\ref{sec:main-results}. In Section~\ref{sec:proofs-asym-independence} we give proofs of the asymptotic independence results and the convergence of job delay. In Section~\ref{sec:independence-upper} we prove the independence upper bound. In Section~\ref{sec:simulations} we provide simulation evaluation of our analysis. We conclude our paper in Section~\ref{sec:conclusions}.

\section{Related Work}\label{sec:related-work}
In this section we discuss prior work on the limited fork-join model and some other related models.  Prior work on the limited fork-join model \cite{RizPolCiu_16,LeeShaHua_17} has focused on the non-asymptotic regime and derived bounds on job delay.  However, the bounds in \cite{RizPolCiu_16,LeeShaHua_17} do not have tightness guarantees.  In particular, the upper bounds there are generally looser than the independence upper bound.  Furthermore, none of the prior work has studied the asymptotic regime of the limited fork-join model. Below we give detailed discussions.

%\paragraph{Limited fork-join model}
\textbf{Limited fork-join model.}
\citet{RizPolCiu_16} give upper bounds on the tail probabilities of job delay in various settings. For Poisson arrivals and exponentially distributed service times, their upper bound is looser than the independence upper bound.  For general service time distributions, their upper bound needs to be computed by numerically solving a non-linear equation.  In contrast, we show that the independence upper bound holds and we also further establish asymptotic tightness of the independence upper bound.

\citet{LeeShaHua_17} give upper and lower bounds on the \emph{mean} job delay, not on the tail probabilities, assuming that service times follow an \emph{exponential} distribution.  Their upper bound is in general looser than the expectation of the independence upper bound, although the difference disappears as $n\to\infty$ when $k\supn=o(n)$.  Compared to this, we prove that the independence upper bound is indeed an upper bound for \emph{any} $n$ and $k\supn$ with $k\supn\le n$.  Besides, we prove it for very general service time distributions and in a stochastic dominance sense, which is stronger than the expectation sense.  Also, there is a gap between their upper and lower bounds and there is no tightness analysis.  Again, we establish asymptotic tightness of the independence upper bound.

There has also been work on variants of the limited fork-join model where each job consists of a random number of tasks. For example, \citet{ShaBouBac_17} simply \emph{assume} that the number of tasks in each job has a distribution such that the numbers of tasks sent to different queues are associated, thus obtaining the independence upper bound for their model. They further investigate different policies for assigning the tasks of a newly arrived job to servers, and show that the job delay under the two studied policies is shorter (in a proper sense) than the job delay under the random assignment in the limited fork-join model. \citet{NelTowTan_88} consider a model where tasks wait in a central queue until some server becomes available. They show that the mean job delay is given by a set of recurrence equations, but no analytical form is derived.  \citet{KumSho_93} obtain upper and lower bounds on the mean delay when tasks are assigned to servers independently.  But still, there are gaps between the upper and lower bounds.
%\vspace{0.5\baselineskip}

%\paragraph{Classic fork-join model}
\textbf{Classic fork-join model.}
The classic fork-join model, where the number of tasks in a job is \emph{equal} to the number of servers, $n$, has been widely studied in the literature.  Similar to the limited fork-join model, tight characterizations of job delay are generally unknown except when $n=2$.  See \cite{Tho_14} for a detailed survey.  Here we just sample several most relevant papers. For general $n$, it has been proven that the mean delay of a job scales as $\Theta(\ln(n))$ as $n\to\infty$ under proper assumptions \cite{NelTan_88,BacMakShw_89}. Besides studying the limited fork-join model, \citet{RizPolCiu_16} also derive an upper bound on the tail distribution of the job delay for the classic fork-join model.  Again, the tightness of the bound is not addressed.
%\vspace{0.5\baselineskip}

%\paragraph{MapReduce}
\textbf{MapReduce.}
Modeling MapReduce systems is challenging since the systems have many complex characteristics such as parallel servers, data locality, communication networks, etc. Most theoretical work on MapReduce does not provide analytical form bounds on the job delay. Papers such as \cite{MosDasKum_11,ZheShrSin_13} and \cite{SunKokShr_17} design scheduling algorithms such that the job delay is guaranteed to be within a constant factor of the optimal, but do not provide analytical bounds. \citet{TanMenZha_12} quantifies the distribution tail of job delay when the map phase is abstracted as a single-server queue, resulting in a system with much higher efficiency, especially when the number of tasks in a job is large.  %\citet{ViaComPon_13} and \citet{FarTooHe_16} derive approximations on the job delay by assuming that the tasks of a job experience independent delays.  But no proof is provided to justify the validity of the approximations.
%\vspace{0.5\baselineskip}

\textbf{Asymptotic task delay.}
One component of the job delay in MapReduce is the task delay. \citet{WanZhuYin_16} and \citet{XieLu_15} bound the mean \emph{task delay}, taking into consideration data locality; however they do not deal with the \emph{job delay}. Bounding \emph{job delay} would require characterizations of the correlation among queues.  \citet{YinSriKan_15} study the task delay in a model where a load-balancing policy called batch-filling is used.  They establish asymptotic independence for a \emph{constant} number of queues, which is insufficient for models with jobs with a growing number of tasks.

\section{Model and Notation}\label{sec:model-notation}

\begin{table}[tb]
\centering
\begin{tabular}{rp{0.35\textwidth}}
\toprule
$n$ & number of servers\\
\midrule
superscript $^{(n)}$ & quantities in the $n$-server system \\
\midrule
$k\supn$ & number of tasks in a job\\
\midrule
$\Lambda\supn$ & job arrival rate\\
\midrule
$\lambda$ & task arrival rate to each queue\\
\midrule
$1/\mu$ & mean of service time\\
\midrule
$\rho$ & load at each queue\\
\midrule
$W\supn_i(t)$ & workload of server $i$'s queue at time $t$\\
\midrule
$T\supn$ & job delay\\
\midrule
$\hat{T}\supn$ & job delay given by independent task delays\\
\midrule
$H_m$ & $m$-th harmonic number: $H_m=\sum_{j=1}^m\frac{1}{j}$\\
\bottomrule
\end{tabular}
  \caption{Notation Table}
  \label{tb:notation}
\end{table}

\textbf{Basic Notation.}
The symbols $\mathbb{R}_+$ and $\mathbb{Z}_+$ denote the set of nonnegative real numbers and nonnegative integers, respectively. We denote random variables by capital letters and vectors by bold letters. When a Markov chain $(\bm{X}(t),t\ge 0)$ has a unique stationary distribution, we denote by $\bm{X}(\infty)$ a random element whose distribution is the stationary distribution.

We denote by $\tod$ convergence in distribution (weak convergence) for random elements. We denote by $d_{TV}(\pi_1,\pi_2)$ the total variation distance between two probability measures $\pi_1$ and $\pi_2$ on a sigma-algebra $\sigma$ of some sample space, i.e.,
\begin{equation}
d_{TV}(\pi_1,\pi_2)=\sup_{\mathcal{S}\in\sigma}|\pi_1(\mathcal{S})-\pi_2(\mathcal{S})|.
\end{equation}

\textbf{Limited fork-join model.}
Our notation is summarized in Table~\ref{tb:notation}.  Recall that we consider a system with $n$ servers, each with its own FIFO queue.  We append the superscript $\supn$ to related quantities to indicate that they are for the $n$-server system. We say that a quantity is a \emph{constant} if it does not scale with $n$.

\emph{Jobs and tasks.}
Jobs arrive over time according to a Poisson process with rate $\Lambda\supn$, and each job consists of $k\supn$ tasks with $k\supn\le n$.  Upon arrival, each job picks $k\supn$ \emph{distinct} servers uniformly at random from the $n$ servers and sends one task to each server. We assume that $\Lambda\supn=n\lambda/k\supn$ for a \emph{constant} $\lambda$, where the constant $\lambda$ is the task arrival rate to each individual queue. Since different jobs choose servers independently, the task arrival process to each queue is also a Poisson process, and the rate is $\lambda$. The service times of tasks are i.i.d.\ following a \cdf\ $G$ with expectation $1/\mu$ and a finite second moment. We think of the service time of each task as being generated upon arrival: each task brings a required service time with it, but the length of the required service time is revealed to the system only when the task is completed.  The load of each queue, $\rho=\lambda/\mu$, is then a constant and we assume that $\rho<1$.

\emph{Queueing dynamics.}
It is not hard to see that each queue is an M/G/1 queue.  But the queues are not independent in general since $k\supn$ tasks arrive to the system at the same time.
Let $W\supn_i(t)$ denote the \emph{workload} of server $i$'s queue at time $t$, i.e., the total remaining service time of all the tasks in the queue, including the partially served task in service.  So the workload of a queue is the waiting time of an incoming task to the queue before the server starts serving it.  Let $\bm{W}\supn(t)=\bigl(W\supn_1(t),W\supn_2(t),\dots,W\supn_n(t)\bigr)$.  Then the workload process, $(\bm{W}\supn(t),t\ge 0)$, is Markovian and ergodic.  The ergodicity can be proven using the rather standard Foster-Lyapunov criteria \cite{MeyTwe_93}, so we omit it here. Therefore, the workload process has a unique stationary distribution and $\bm{W}\supn(t)\tod \bm{W}\supn(\infty)$ as $t\to\infty$.

\emph{Job delay.}
We are interested in the distribution of job delay in steady state, i.e., the delay a job would experience if it arrives to the system and finds the system in steady state.  Let a random variable $T\supn$ represent this steady-state job delay.  Specifically, the distribution of $T\supn$ is determined by the workload $\bm{W}\supn(\infty)$ in the following way.  When a job comes into the system, its tasks are sent to $k\supn$ queues and experience the delays in these queues.  Since the queueing processes are symmetric over the indices of queues, without loss of generality, we can assume that the tasks are sent to the first $k\supn$ queues for the purpose of computing the distribution of $T\supn$.  The delay of a task is the sum of its waiting time and service time.  So the task delay in queue $i$, denoted by $T\supn_i$, can be written as $T\supn_i=W\supn_i(\infty)+X_i$ with $X_i$ being the service time.  Recall that the $X_i$'s are i.i.d.$\sim G$ and independent of everything else.  Since the job is completed only when all its tasks are completed,
\begin{equation}\label{eq:T}
T\supn = \max\left\{T\supn_1,T\supn_2,\dots,T\supn_{k\supn}\right\}.
\end{equation}

We will study the relation between $T\supn$ and $\hat{T}\supn$ with $\hat{T}\supn$ defined as the job delay given by \emph{independent} task delays.  Specifically, $\hat{T}\supn$ can be expressed as:
\begin{equation}\label{eq:That}
\hat{T}\supn=\max\left\{\hat{T}\supn_1,\hat{T}\supn_2,\dots,\hat{T}\supn_{k\supn}\right\},
\end{equation}
where $\hat{T}\supn_1,\hat{T}\supn_2,\dots,\hat{T}\supn_{k\supn}$ are i.i.d.\ and each $\hat{T}\supn_i$ has the same distribution as $T\supn_i$. 
Again, due to symmetry, all the $T\supn_i$'s have the same distribution. Let $F$ denote the \cdf\ of $T\supn_i$, whose form is known from the queueing theory literature.  Then, we have the following explicit form for $\hat{T}\supn$:
\begin{equation}
\Pr\left(\hat{T}\supn\le \tau\right)=\left(F(\tau)\right)^{k\supn},\quad \tau\ge 0.
\end{equation}

\section{Main Results}\label{sec:main-results}
In Theorem~\ref{thm:asym-independence}, we establish asymptotic independence of any $k\supn$ queues under the condition $k\supn = o(n^{1/4})$ as the number of servers $n\to\infty$. The asymptotic independence is in the sense that the total variation distance between the distribution of the workloads of these queues and the distribution of $k\supn$ independent queues goes to $0$ as $n\to\infty$.  Consequently, the distance between the distribution of job delay, $T\supn$, and the distribution of the job delay given by independent task delays, $\hat{T}\supn$, goes to $0$. This result indicates that assuming independence among the delays of a job's tasks gives a good approximation of job delay when the system is large.  Again, due to symmetry, we can focus on the first $k\supn$ queues without loss of generality.

\begin{theorem}\label{thm:asym-independence}
\begin{sloppypar}
Consider an $n$-server system in the limited fork-join model with $k\supn=o(n^{1/4})$. Let $\pi^{(n,k\supn)}$ denote the joint distribution of the steady-state workloads $W\supn_1(\infty)$, $W\supn_2(\infty),\dots,W\supn_{k\supn}(\infty)$, and $\hat{\pi}^{(k\supn)}$ denote the product distribution of $k\supn$ i.i.d.\ random variables, each of which follows a distribution that is the same as the distribution of $W\supn_1(\infty)$. Then
\begin{equation}\label{eq:asym-independence}
\lim_{n\to\infty}d_{TV}\Bigl(\pi^{(n,k\supn)},\hat{\pi}^{(k\supn)}\Bigr)=0.
\end{equation}
Consequently, the steady-state job delay, $T\supn$, and the job delay given by independent task delays as defined in \eqref{eq:That}, $\hat{T}\supn$, satisfy
\begin{equation}\label{eq:converge-cdf-general}
\lim_{n\to\infty}\sup_{\tau\ge 0}\left|\Pr\bigl(T\supn\le \tau\bigr)-\Pr\bigl(\hat{T}\supn\le \tau\bigr)\right|=0.
\end{equation}
\end{sloppypar}
\end{theorem}

For the special case where the service times are exponentially distributed, the job delay asymptotics have explicit forms presented in Corollary~\ref{cor:exp} below.
\begin{corollary}\label{cor:exp}
Consider an $n$-server system in the limited fork-join model with $k\supn=o(n^{1/4})$, job arrival rate $\Lambda\supn=n\lambda/k\supn$, and exponentially distributed service times with mean $1/\mu$.  Then the steady-state job delay, $T\supn$, converges as:
\begin{equation}\label{eq:converge-cdf}
\lim_{n\to\infty}\sup_{\tau\ge 0}\left|\Pr\bigl(T\supn\le \tau\bigr)-\left(1-e^{-(\mu-\lambda)\tau}\right)^{k\supn}\right|=0,
\end{equation}
Specifically, if $k\supn\to\infty$ as $n\to\infty$, then
\begin{equation}\label{eq:converge-in-distr}
\frac{T\supn}{H_{k\supn}/(\mu-\lambda)}\tod 1,\quad\text{as }n\to\infty,
\end{equation}
where $H_{k\supn}$ is the $k\supn$-th harmonic number, and further,
\begin{equation}\label{eq:converge-expectation}
\lim_{n\to\infty}\frac{\expect\bigl[T\supn\bigr]}{H_{k\supn}/(\mu-\lambda)}=1.
\end{equation}
\end{corollary}

The results above characterize job delay in the asymptotic regime where $n$ goes to infinity.  In Theorem~\ref{thm:independence-upper} below, we study the non-asymptotic regime for any $n$ and any $k\supn$ with $k\supn=k\le n$, and we establish the independence upper bound on job delay.

\begin{theorem}\label{thm:independence-upper}
Consider an $n$-server system in the limited fork-join model with $k\supn=k\le n$. Then the steady-state job delay, $T\supn$, is stochastically upper bounded by the job delay given by independent task delays as defined in \eqref{eq:That}, $\hat{T}\supn$, i.e.,
\begin{equation}\label{eq:ind-upper}
T\supn\le_{st} \hat{T}\supn,
\end{equation}
where ``$\le_{st}$'' denotes stochastic dominance. Specifically, for any $\tau\ge 0$,
\begin{align}
\Pr\bigl(T\supn > \tau\bigr)&\le\Pr\bigl(\hat{T}\supn > \tau\bigr) = 1-\left(F(\tau)\right)^{k\supn}.\label{eq:ind-upper-tail}
\end{align}
\end{theorem}

\section{Proofs of Asymptotic Independence and Job Delay Asymptotics}\label{sec:proofs-asym-independence}

In this section, we prove the asymptotic independence and job delay asymptotics in Theorem~\ref{thm:asym-independence} and Corollary~\ref{cor:exp}.

%\begin{namedthm*}{Theorem~\ref{thm:asym-independence} (Restated)}
%\begin{sloppypar}
%Consider an $n$-server system in the limited fork-join model with $k\supn=o(n^{1/4})$. Let $\pi^{(n,k\supn)}$ denote the joint distribution of the steady-state workloads $W\supn_1(\infty)$, $W\supn_2(\infty),\dots,W\supn_{k\supn}(\infty)$, and $\hat{\pi}^{(k\supn)}$ denote the product distribution of $k\supn$ i.i.d.\ random variables, each of which follows a distribution that is the same as the distribution of $W\supn_1(\infty)$. Then
%\begin{equation}\tag*{(\ref{eq:asym-independence}) (restated)}
%\lim_{n\to\infty}d_{TV}\Bigl(\pi^{(n,k\supn)},\hat{\pi}^{(k\supn)}\Bigr)=0.
%\end{equation}
%Consequently, the steady-state job delay, $T\supn$, and the job delay given by independent task delays as defined in \eqref{eq:That}, $\hat{T}\supn$, satisfy that
%\begin{equation}\tag*{(\ref{eq:converge-cdf-general}) (restated)}
%\lim_{n\to\infty}\sup_{\tau\ge 0}\left|\Pr\bigl(T\supn\le \tau\bigr)-\Pr\bigl(\hat{T}\supn\le \tau\bigr)\right|=0.
%\end{equation}
%\end{sloppypar}
%\end{namedthm*}

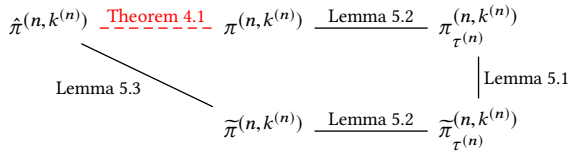
\begin{figure}
\centering
\begin{tikzcd}[column sep=huge]
\hat{\pi}^{(n,k\supn)}
\arrow[r, dash, dashed, red, "\text{Theorem~\ref{thm:asym-independence}}"]
\arrow[dr,dash,"\text{Lemma~\ref{LEM:PI-TILDE-VS-PI-HAT}}"']&
\pi^{(n,k\supn)}
\arrow[r, dash, "\text{Lemma~\ref{lem:convergence-rate}}"]&
\pi^{(n,k\supn)}_{\tau\supn}
\arrow[d, dash, "\text{Lemma~\ref{lem:independence-finite-time}}"]\\
&
\widetilde{\pi}^{(n,k\supn)}
\arrow[r, dash, "\text{Lemma~\ref{lem:convergence-rate}}"]&
\widetilde{\pi}^{(n,k\supn)}_{\tau\supn}
\end{tikzcd}
\caption{Distances in the proof of Theorem~\ref{thm:asym-independence}}
\label{fig:distances}
\end{figure}
\textbf{Proof Sketch.}
To prove Theorem~\ref{thm:asym-independence}, we couple each $n$-server system in the limited fork-join model, which we refer to as system $\mathcal{S}\supn$, with a system $\widetilde{\mathcal{S}}\supn$ in which the first $k\supn$ queues are \emph{independent}. We will specify $\widetilde{\mathcal{S}}\supn$ below.  Let $\widetilde{W}\supn_i(t)$ denote the workload of server $i$ at time $t$ in system $\widetilde{\mathcal{S}}\supn$. Let $\widetilde{\bm{W}}^{(n,k\supn)}(t)=\left(\widetilde{W}\supn_1(t),\dots,\widetilde{W}\supn_{k\supn}(t)\right)$ and $\bm{W}^{(n,k\supn)}(t)=\left(W\supn_1(t),\dots,W\supn_{k\supn}(t)\right)$. Then the proof will proceed in the following three steps, where we break down the distance $d_{TV}\left(\pi^{(n,k\supn)},\hat{\pi}^{(k\supn)}\right)$ in Theorem~\ref{thm:asym-independence} into three parts, illustrated in Figure~\ref{fig:distances}.

(i) We carefully choose a finite time $\tau\supn$ and consider systems $\mathcal{S}\supn$ and $\widetilde{\mathcal{S}}\supn$ at time $\tau\supn$.  We show in Lemma~\ref{lem:independence-finite-time} that the distribution of $\bm{W}^{(n,k\supn)}(\tau\supn)$, denoted by $\pi_{\tau\supn}^{(n,k\supn)}$, approaches the distribution of $\widetilde{\bm{W}}^{(n,k\supn)}(\tau\supn)$, denoted by $\widetilde{\pi}_{\tau\supn}^{(n,k\supn)}$, as $n\to\infty$.

(ii) We show in Lemma~\ref{lem:convergence-rate} that in both systems $\mathcal{S}\supn$ and $\widetilde{\mathcal{S}}\supn$, the finite-time distributions $\pi_{\tau\supn}^{(n,k\supn)}$ and $\widetilde{\pi}_{\tau\supn}^{(n,k\supn)}$ are close to the stationary distributions, $\pi^{(n,k\supn)}$ and $\widetilde{\pi}^{(n,k\supn)}$, respectively.

(iii) We show in Lemma~\ref{LEM:PI-TILDE-VS-PI-HAT} that the stationary distribution $\widetilde{\pi}^{(n,k\supn)}$ in system $\widetilde{\mathcal{S}}\supn$ is close to the product distribution $\hat{\pi}^{(k\supn)}$ in Theorem~\ref{thm:asym-independence}. Note that both $\widetilde{\pi}^{(n,k\supn)}$ and $\hat{\pi}^{(k\supn)}$ are for $k\supn$ \emph{independent} workloads, but we will see that their loads are different.

\textbf{Coupling.}
Now we specify the coupling between $\mathcal{S}\supn$ and $\widetilde{\mathcal{S}}\supn$.  Both systems have $n$ servers and the queues are all empty at time $0$, i.e., $W\supn_i(0)=\widetilde{W}\supn_i(0)=0$ for all $i=1,\dots,n$.  When there is a job arrival to system $\mathcal{S}\supn$, we let a job also arrive to system $\widetilde{\mathcal{S}}\supn$.  Recall that the job arrival in $\mathcal{S}\supn$ selects $k\supn$ distinct queues uniformly at random and sends one task to each queue. If it selects \emph{at most one} queue from the set $\bigl\{1,2,\dots,k\supn\bigr\}$, then we let the job arrival in $\widetilde{\mathcal{S}}\supn$ send its tasks to queues with the same indices as those in $\mathcal{S}\supn$.  Otherwise, suppose it selects queues $i_1,i_2,\dots,i_m$ from $\bigl\{1,2,\dots,k\supn\bigr\}$ with $2\le m\le k\supn$.  Then we let the job arrival in $\widetilde{\mathcal{S}}\supn$ send one task to a queue chosen uniformly at random from $i_1,i_2,\dots,i_m$, kill the other $m-1$ tasks, and send the remaining $k\supn-m$ tasks to queues with the same indices as those in $\mathcal{S}\supn$.  For each pair of tasks in $\mathcal{S}\supn$ and $\widetilde{\mathcal{S}}\supn$ that are sent to queues with the same indices, we let them have the same service time.

It can be verified that in system $\widetilde{\mathcal{S}}\supn$, the queues $1,2,\dots,k\supn$ are independent M/G/1 queues with arrival rate $\widetilde{\lambda}\supn$ and mean service time $1/\mu$, where
\begin{equation}\label{eq:lambda-tilde}
\widetilde{\lambda}\supn=\frac{\Lambda\supn}{k\supn}\Biggl(1-\frac{\binom{n-k\supn}{k\supn}}{\binom{n}{k\supn}}\Biggr).
\end{equation}
Let $\widetilde{\rho}\supn=\frac{\widetilde{\lambda}\supn}{\mu}$ denote the load of each queue. Note that $\widetilde{\lambda}\supn < \lambda$ but $\widetilde{\lambda}\supn\to\lambda$ as $n\to\infty$.  Specifically,
\begin{equation*}
\lambda-\widetilde{\lambda}\supn=O\biggl(\frac{(k\supn)^2}{n}\biggr).
\end{equation*}

\subsection{Lemmas Needed for Theorem~\ref{thm:asym-independence}}
We first show in Lemma~\ref{lem:independence-finite-time} that, over a finite time interval with proper length, any $k\supn$ queues in the $n$-server system $\mathcal{S}\supn$ are asymptotically independent as the number of servers $n\to\infty$.
\begin{lemma}\label{lem:independence-finite-time}
For any time $\tau\supn$ with $\tau\supn=O\Bigl(\frac{n^{1/2}}{k\supn}\Bigr)$,
\begin{equation}\label{eq:independence-finite-time}
d_{TV}\Bigl(\pi_{\tau\supn}^{(n,k\supn)},\widetilde{\pi}_{\tau\supn}^{(n,k\supn)}\Bigr)=O\Biggl(\biggl(\frac{k\supn}{n^{1/4}}\biggr)^2\Biggr),
\end{equation}
which goes to $0$ as $n\to\infty$.
\end{lemma}

Lemma~\ref{lem:convergence-rate} states that the time interval $\tau\supn$ in Lemma~\ref{lem:independence-finite-time} is long enough for the systems $\mathcal{S}\supn$ and $\widetilde{\mathcal{S}}\supn$ to be close to steady state.
\begin{lemma}\label{lem:convergence-rate}
For any time $\tau\supn$ with $\tau\supn=O\Bigl(\frac{n^{1/2}}{k\supn}\Bigr)$,
\begin{equation}\label{eq:convergence}
d_{TV}\Bigl(\pi_{\tau\supn}^{(n,k\supn)},\pi^{(n,k\supn)}\Bigr)=O\Biggl(\biggl(\frac{k\supn}{n^{1/4}}\biggr)^2\Biggr),
\end{equation}
and
\begin{equation}\label{eq:convergence-independent}
d_{TV}\Bigl(\widetilde{\pi}_{\tau\supn}^{(n,k\supn)},\widetilde{\pi}^{(n,k\supn)}\Bigr)=O\Biggl(\biggl(\frac{k\supn}{n^{1/4}}\biggr)^2\Biggr).
\end{equation}
\end{lemma}

The distribution $\widetilde{\pi}^{(n,k\supn)}$ is the joint distribution of the steady-state workloads of $k\supn$ independent queues, each with arrival rate $\widetilde{\lambda}\supn$ and mean service time $1/\mu$.  Since $\widetilde{\lambda}\supn\to\lambda$ as $n\to\infty$, $\widetilde{\pi}^{(n,k\supn)}$ approaches the product distribution $\hat{\pi}^{(k\supn)}$ in Theorem~\ref{thm:asym-independence}, which is for $k\supn$ independent queues each with arrival rate $\lambda$ and mean service time $1/\mu$.  This is formally stated in Lemma~\ref{LEM:PI-TILDE-VS-PI-HAT}.
\begin{lemma}\label{LEM:PI-TILDE-VS-PI-HAT}
\begin{equation}\label{EQ:PI-TILDE-VS-PI-HAT}
d_{TV}\Bigl(\widetilde{\pi}^{(n,k\supn)},\hat{\pi}^{(k\supn)}\Bigr)=O\Biggl(\biggl(\frac{k\supn}{n^{1/4}}\biggr)^2\Biggr).
\end{equation}
\end{lemma}

\subsection{Proof of Theorem~\ref{thm:asym-independence} Given Lemmas}
\begin{proof}
The proof of the asymptotic independence in \eqref{eq:asym-independence} in Theorem~\ref{thm:asym-independence} is straightforward given the lemmas.  Pick any $\tau\supn$ with $\tau\supn=O\Bigl(\frac{n^{1/2}}{k\supn}\Bigr)$.  Then
\begin{align*}
&\mspace{23mu}d_{TV}\Bigl(\pi^{(n,k\supn)},\hat{\pi}^{(k\supn)}\Bigr)\\
&\le d_{TV}\Bigl(\pi^{(n,k\supn)},\pi_{\tau\supn}^{(n,k\supn)}\Bigr)+d_{TV}\Bigl(\pi_{\tau\supn}^{(n,k\supn)},\widetilde{\pi}_{\tau\supn}^{(n,k\supn)}\Bigr)\\
&\mspace{23mu}+d_{TV}\Bigl(\widetilde{\pi}_{\tau\supn}^{(n,k\supn)},\widetilde{\pi}^{(n,k\supn)}\Bigr)+d_{TV}\Bigl(\widetilde{\pi}^{(n,k\supn)},\hat{\pi}^{(k\supn)}\Bigr)\\
&=O\Biggl(\biggl(\frac{k\supn}{n^{1/4}}\biggr)^2\Biggr).
\end{align*}
Therefore, under the condition that $k\supn=o(n^{1/4})$, there holds
\begin{equation*}
\lim_{n\to\infty}d_{TV}\Bigl(\pi^{(n,k\supn)},\hat{\pi}^{(k\supn)}\Bigr)=0.
\end{equation*}

Next we prove the job delay asymptotics in \eqref{eq:converge-cdf-general} in Theorem~\ref{thm:asym-independence}.  Recall that $\pi^{(n,k\supn)}$ and $\hat{\pi}^{(k\supn)}$ are distributions of workloads.  Below we compute the distributions of $T\supn$ and $\hat{T}\supn$ using $\pi^{(n,k\supn)}$ and $\hat{\pi}^{(k\supn)}$, which allows us to bound the distance between the distributions of $T\supn$ and $\hat{T}\supn$ using $d_{TV}\left(\pi^{(n,k\supn)},\hat{\pi}^{(k\supn)}\right)$.  By the representations of $T\supn$ and $\hat{T}\supn$ in \eqref{eq:T} and \eqref{eq:That}, we have that for any $\tau\ge 0$,
\begin{align*}
&\mspace{23mu}\bigl|\Pr\bigl(T\supn\le \tau\bigr)-\Pr\bigl(\hat{T}\supn\le \tau\bigr)\bigr|\\
&=\Biggl|\int_{\bm{w}\in\mathbb{R}_+^{k\supn}}\Biggl(\prod_{i=1}^{k\supn}\Pr(w_i+X_i\le\tau)\Biggr)d\pi^{(n,k\supn)}(\bm{w})\\
&\mspace{24mu}-\int_{\bm{w}\in\mathbb{R}_+^{k\supn}}\Biggl(\prod_{i=1}^{k\supn}\Pr(w_i+X_i\le\tau)\Biggr)d\hat{\pi}^{(k\supn)}(\bm{w})\Biggr|\\
&\le \int_{\bm{w}\in\mathbb{R}_+^{k\supn}}\Biggl(\prod_{i=1}^{k\supn}\Pr(w_i+X_i\le\tau)\Biggr)\cdot\bigl|d\pi^{(n,k\supn)}(\bm{w})-d\hat{\pi}^{(k\supn)}(\bm{w})\bigr|\\
&\le \int_{\bm{w}\in\mathbb{R}_+^{k\supn}}\bigl|d\pi^{(n,k\supn)}(\bm{w})-d\hat{\pi}^{(k\supn)}(\bm{w})\bigr|\\
&=2 d_{TV}\Bigl(\pi^{(n,k\supn)},\hat{\pi}^{(k\supn)}\Bigr).
\end{align*}
Therefore,
\begin{align*}
&\mspace{23mu}\lim_{n\to\infty}\sup_{\tau\ge 0} \left|\Pr(T\supn\le\tau)-\Pr\bigl(\hat{T}\supn\le \tau\bigr)\right|\\
&\le \lim_{n\to\infty}2 d_{TV}\Bigl(\pi^{(n,k\supn)},\hat{\pi}^{(k\supn)}\Bigr)\\
&=0.
\end{align*}
\end{proof}

\subsection{Proof of Lemmas}
\subsection*{\normalsize{Proof of Lemma~\ref{lem:independence-finite-time}}}
\begin{proof}
In order to bound $d_{TV}\Bigl(\pi_{\tau\supn}^{(n,k\supn)},\widetilde{\pi}_{\tau\supn}^{(n,k\supn)}\Bigr)$, we first write
\begin{align*}
&\mspace{25mu}d_{TV}\Bigl(\pi_{\tau\supn}^{(n,k\supn)},\widetilde{\pi}_{\tau\supn}^{(n,k\supn)}\Bigr)\\
&\le\Pr\Bigl(\bm{W}^{(n,k\supn)}(\tau\supn)\neq\widetilde{\bm{W}}^{(n,k\supn)}(\tau\supn)\Bigr)\\
&\le \Pr\Bigl(\bm{W}^{(n,k\supn)}(t)\neq\widetilde{\bm{W}}^{(n,k\supn)}(t)\text{ for some }t\in[0,\tau\supn]\Bigr),
\end{align*}
where the first inequality follows from a standard property of total variation distance.
By the coupling between $\mathcal{S}\supn$ and $\widetilde{\mathcal{S}}\supn$, $\bm{W}^{(n,k\supn)}(t)$ and $\widetilde{\bm{W}}^{(n,k\supn)}(t)$ are different for some time $t\in[0,\tau\supn]$ only when at least one job arrival during $[0,\tau\supn]$ selects more than one queue from $\bigl\{1,\dots,k\supn\bigr\}$ in system $\mathcal{S}\supn$.  We denote this event by $\mathcal{E}$.  Then
\begin{equation*}
d_{TV}\Bigl(\pi_{\tau\supn}^{(n,k\supn)},\widetilde{\pi}_{\tau\supn}^{(n,k\supn)}\Bigr)\le \Pr(\mathcal{E}).
\end{equation*}
So it suffices to prove that
\begin{align*}
\Pr(\mathcal{E})
&=O\Biggl(\biggl(\frac{k\supn}{n^{1/4}}\biggr)^2\Biggr)
\end{align*}
for $\tau\supn$ with $\tau\supn=O\Bigl(\frac{n^{1/2}}{k\supn}\Bigr)$.  The remainder of this proof is dedicated to bounding $\Pr(\mathcal{E})$.

Let $p\supn$ denote the probability for a job arrival to select less than or equal to $1$ queue from queues $1,2,\dots,k\supn$ in system $\mathcal{S}\supn$. Then
\begin{equation*}
p\supn=\frac{\binom{n-k\supn}{k\supn}}{\binom{n}{k\supn}}+\frac{k\supn\binom{n-k\supn}{k\supn-1}}{\binom{n}{k\supn}}.
\end{equation*}
Let $A$ be the number of job arrivals during $[0,\tau\supn]$. Then
\begin{align}
\Pr(\mathcal{E})
&=\sum_{j=0}^{\infty}\Pr(A=j)\Pr(\mathcal{E}\mid A=j)\nonumber\\
&=\sum_{j=0}^{\infty}\frac{(\Lambda\supn\tau\supn
)^je^{-\Lambda\supn\tau\supn}}{j!}\Bigl(1-\bigl(p\supn\bigr)^j\Bigr)\nonumber\\
&=1-e^{-\Lambda\supn\tau\supn(1-p\supn)},\label{eq:PrE}
\end{align}
where \eqref{eq:PrE} follows from the definition of the Poisson generating function.  We calculate $p\supn$ as follows:
\begin{align*}
p\supn&=\frac{\binom{n-k\supn}{k\supn}}{\binom{n}{k\supn}}+\frac{k\supn\binom{n-k\supn}{k\supn-1}}{\binom{n}{k\supn}}\\
&=\frac{(n-k\supn)!}{(n-2k\supn)!}\frac{(n-k\supn)!}{n!}\biggl(1+\frac{(k\supn)^2}{n-2k\supn+1}\biggr)\\
&=\biggl(1-\frac{k\supn}{n}\biggr)\biggl(1-\frac{k\supn}{n-1}\biggr)\dots\biggl(1-\frac{k\supn}{n-k\supn+1}\biggr)\\
&\mspace{21mu}\cdot\biggl(1+\frac{(k\supn)^2}{n-2k\supn+1}\biggr)\\
&\ge \biggl(1-\frac{k\supn}{n-k\supn+1}\biggr)^{k\supn}\biggl(1+\frac{(k\supn)^2}{n-k\supn+1}\biggr).
\end{align*}
Since
\begin{align*}
&\mspace{21mu}\biggl(1-\frac{k\supn}{n-k\supn+1}\biggr)^{k\supn}\\
&=e^{k\supn\ln\bigl(1-\frac{k\supn}{n-k\supn+1}\bigr)}\\
&=e^{-\frac{(k\supn)^2}{n-k\supn+1}+O\bigl(\frac{(k\supn)^3}{(n-k\supn+1)^2}\bigr)}\\
&=1-\frac{(k\supn)^2}{n-k\supn+1}+O\biggl(\frac{(k\supn)^4}{(n-k\supn+1)^2}\biggr),
\end{align*}
we have
\begin{align*}
1-p\supn&\le 1-\biggl(1-\frac{k\supn}{n-k\supn+1}\biggr)^{k\supn}\biggl(1+\frac{(k\supn)^2}{n-k\supn+1}\biggr)\\
&=O\biggl(\frac{(k\supn)^4}{(n-k\supn+1)^2}\biggr).
\end{align*}
Recall that
\begin{equation*}
\Lambda\supn k\supn=n\lambda,\quad \tau\supn=O\left(\frac{n^{1/2}}{k\supn}\right).
\end{equation*}
Thus,
\begin{align*}
\Lambda\supn\tau\supn(1-p\supn)%&\le  \Lambda\supn O(k\supn)O\biggl(\frac{(k\supn)^4}{(n-k\supn+1)^2}\biggr)\\
&=O\Biggl(\biggl(\frac{k\supn}{n^{1/4}}\biggr)^2\Biggr).
\end{align*}
Consequently, inserting this to \eqref{eq:PrE} yields
\begin{align*}
\Pr(\mathcal{E})
&=O\Biggl(\biggl(\frac{k\supn}{n^{1/4}}\biggr)^2\Biggr),
\end{align*}
which completes the proof of Lemma~\ref{lem:independence-finite-time}.
\end{proof}

\subsection*{\normalsize{Proof of Lemma~\ref{lem:convergence-rate}}}
\begin{proof}
We first prove \eqref{eq:convergence}. For system $\mathcal{S}\supn$, we consider the following coupling between two copies of the workload process. In one copy, the system starts from empty queues, i.e., this is the workload process $(\bm{W}\supn(t),t\ge 0)$ we have introduced. The other copy, which we denote by $(\overline{\bm{W}}\supn(t),t\ge 0)$, starts from its stationary distribution, i.e., the distribution of $\overline{\bm{W}}\supn(0)$ is $\pi^{(n)}$. Then the distribution of $\overline{\bm{W}}\supn(t)$ is $\pi^{(n)}$ for any $t$. We let these two copies have the same arrival processes, and each arriving task has the same service time under both queueing processes.

By this coupling, $W\supn_i(t)\le \overline{W}\supn_i(t)$ for any time $t$ and any~$i$. For each $i=1,2,\dots,k\supn$, let $\tau\supn_i$ be the earliest time that the workload $\overline{W}\supn_i$ is $0$, i.e.,
\begin{equation*}
\tau\supn_i=\min\Bigl\{t\colon \overline{W}\supn_i(u)=0\text{ for some }u\in[0,t]\Bigr\}.
\end{equation*}
Let
\begin{equation*}
\tau\supn_0=\max\Bigl\{\tau\supn_1,\dots,\tau\supn_{k\supn}\Bigr\}.
\end{equation*}
Then $W\supn_i(t)= \overline{W}\supn_i(t)$ for any $t\ge \tau\supn_0$ and any $i=1,2,\dots,k\supn$.

To show \eqref{eq:convergence} in Lemma~\ref{lem:convergence-rate}, which we restate here for reference
\begin{equation}\tag*{(\ref{eq:convergence}) (Restated)}
d_{TV}\Bigl(\pi_{\tau\supn}^{(n,k\supn)},\pi^{(n,k\supn)}\Bigr)=O\Biggl(\biggl(\frac{k\supn}{n^{1/4}}\biggr)^2\Biggr),
\end{equation}
it suffices to prove that
\begin{equation}\label{eq:sufficient}
\Pr(\tau\supn_0>\tau\supn)=O\Biggl(\biggl(\frac{k\supn}{n^{1/4}}\biggr)^2\Biggr).
\end{equation}
To see that this is sufficient, we first note that
\begin{align}
&\mspace{23mu}d_{TV}\Bigl(\pi_{\tau\supn}^{(n,k\supn)},\pi^{(n,k\supn)}\Bigr)\nonumber\\
&\le \Pr\Bigl(\bm{W}^{(n,k\supn)}(\tau\supn)\neq\overline{\bm{W}}^{(n,k\supn)}(\tau\supn)\Bigr).\label{eq:dtv}
\end{align}
By the definition of $\tau\supn_0$, $\bm{W}^{(n,k\supn)}(\tau\supn)\neq\overline{\bm{W}}^{(n,k\supn)}(\tau\supn)$ if and only if $\tau\supn_0>\tau\supn$. So \eqref{eq:dtv} further implies
\begin{align}
d_{TV}\Bigl(\pi_{\tau\supn}^{(n,k\supn)},\pi^{(n,k\supn)}\Bigr)\le \Pr(\tau\supn_0>\tau\supn),
\end{align}
and thus \eqref{eq:sufficient} implies \eqref{eq:convergence}.

Now we prove \eqref{eq:sufficient}.  Note that the distribution of $\tau\supn_i$ does not depend on $n$ since each individual $(\overline{W}\supn_i(t),t\ge 0)$ evolves as an M/G/1 queue with arrival rate $\lambda$ and service time distribution $G$, and $\tau\supn_i$ is a busy period started by the amount of work in steady state. Thus, by standard results on busy periods (see, e.g., \cite{Har_13}),
\begin{equation}\label{eq:busy}
\expect[\tau\supn_i]=\frac{\lambda g_2}{2(1-\rho)^2},
\end{equation}
where $g_2$ is the second moment of $G$, which is a constant.  By Markov's inequality,
\begin{equation}
\Pr(\tau\supn_i>\tau\supn)\le \frac{\expect[\tau\supn_i]}{\tau\supn}.
\end{equation}
Since $\tau\supn_0=\max\{\tau\supn_1,\tau\supn_2,\dots,\tau\supn_{k\supn}\}$, by the union bound we have
\begin{align*}
\Pr(\tau\supn_0 > \tau\supn)&\le \sum_{i=1}^{k\supn}\Pr(\tau\supn_i > \tau\supn)\\
&\le \frac{k\supn}{\tau\supn}\frac{\lambda g_2}{2(1-\rho)^2}\\
&=O\Biggl(\biggl(\frac{k\supn}{n^{1/4}}\biggr)^2\Biggr).
\end{align*}
This is \eqref{eq:sufficient} and thus it completes the proof of \eqref{eq:convergence}.

The proof of \eqref{eq:convergence-independent} in Lemma~\ref{lem:convergence-rate} is very much similar to the proof of \eqref{eq:convergence}. We obtain \eqref{eq:convergence-independent} by noting that each $(\widetilde{W}\supn_i(t),t\ge 0)$ with $i=1,2,\dots,k\supn$ is an M/G/1 queue with arrival rate $\widetilde{\lambda}\supn<\lambda$ and following arguments similar to those in the proof of \eqref{eq:convergence}.
\end{proof}

\subsection*{\normalsize{Proof of Lemma~\ref{LEM:PI-TILDE-VS-PI-HAT}}}
\begin{proof}
This proof has a similar flavor to the proofs of Lemmas~\ref{lem:independence-finite-time} and \ref{lem:convergence-rate}. Recall that $\left(\widetilde{\bm{W}}^{(n,k\supn)}(t),t\ge 0\right)$, the workload processes of the first $k\supn$ queues in system $\widetilde{\mathcal{S}}\supn$, are $k\supn$ independent M/G/1 queues each with arrival rate $\widetilde{\lambda}\supn$ and service time distribution $G$.  We couple this with $\left(\hat{\bm{W}}^{(k\supn)}(t),t\ge 0\right)$, where $\hat{\bm{W}}^{(k\supn)}(t)=\bigl(\hat{W}_1(t),\dots,\hat{W}_{k\supn}(t)\bigr)$ is the workload vector of $k\supn$ independent M/G/1 queues each with arrival rate $\lambda$ and service time distribution $G$. Then $\hat{\pi}^{(k\supn)}$ is its stationary distribution.  We will prove the bound on  $d_{TV}\bigl(\widetilde{\pi}^{(n,k\supn)},\hat{\pi}^{(k\supn)}\bigr)$ by showing that $\left(\widetilde{\bm{W}}^{(n,k\supn)}(t),t\ge 0\right)$ and $\left(\hat{\bm{W}}^{(k\supn)}(t),t\ge 0\right)$ are close.

Now we specify the coupling.  All the queues start from empty, i.e., $\hat{W}_i(0)=\widetilde{W}\supn_i(0)=0$ for all $i=1,2,\dots,k\supn$. When there is a task arrival to some queue of $\hat{\bm{W}}^{(k\supn)}$, we let a task arrive to the corresponding queue of $\widetilde{\bm{W}}^{(n,k\supn)}$ with probability $\frac{\widetilde{\lambda}\supn}{\lambda}$, and let these two tasks require the same service time.  So with probability $1-\frac{\widetilde{\lambda}\supn}{\lambda}$ there is no task arrival to $\widetilde{\bm{W}}^{(n,k\supn)}$.

We pick a time $\tau\supn=O\Bigl(\frac{n^{1/2}}{k\supn}\Bigr)$. Let $\hat{\pi}^{(k\supn)}_{\tau\supn}$ denote the distribution of $\hat{\bm{W}}^{(k\supn)}(\tau\supn)$.  Then
\begin{multline*}
d_{TV}\Bigl(\widetilde{\pi}^{(n,k\supn)},\hat{\pi}^{(k\supn)}\Bigr)\le d_{TV}\Bigl(\widetilde{\pi}^{(n,k\supn)}_{\tau\supn},\hat{\pi}^{(k\supn)}_{\tau\supn}\Bigr)\\
+ d_{TV}\Bigl(\widetilde{\pi}^{(n,k\supn)}_{\tau\supn},\widetilde{\pi}^{(n,k\supn)}\Bigr)+d_{TV}\Bigl(\hat{\pi}^{(k\supn)}_{\tau\supn},\hat{\pi}^{(k\supn)}\Bigr).
\end{multline*}
Noting Lemma~\ref{lem:convergence-rate}, we have
\begin{align}
d_{TV}\Bigl(\widetilde{\pi}_{\tau\supn}^{(n,k\supn)},\widetilde{\pi}^{(n,k\supn)}\Bigr)&=O\Biggl(\biggl(\frac{k\supn}{n^{1/4}}\biggr)^2\Biggr),\label{eq:dis1}\\
d_{TV}\Bigl(\hat{\pi}_{\tau\supn}^{(k\supn)},\hat{\pi}^{(k\supn)}\Bigr)&=O\Biggl(\biggl(\frac{k\supn}{n^{1/4}}\biggr)^2\Biggr).\label{eq:dis2}
\end{align}
Next we bound $d_{TV}\Bigl(\widetilde{\pi}^{(n,k\supn)}_{\tau\supn},\hat{\pi}^{(k\supn)}_{\tau\supn}\Bigr)$ using arguments similar to those in the proof of Lemma~\ref{lem:independence-finite-time}. By the coupling, $\hat{\bm{W}}^{(k\supn)}(t)$ and $\widetilde{\bm{W}}^{(n,k\supn)}(t)$ are different for some $t\in[0,\tau\supn]$ only when some task arrives to $\hat{\bm{W}}^{(k\supn)}$ but not to $\widetilde{\bm{W}}^{(n,k\supn)}$. We denote this event by $\mathcal{E}$. Then
\begin{equation*}
d_{TV}\Bigl(\widetilde{\pi}^{(n,k\supn)}_{\tau\supn},\hat{\pi}^{(k\supn)}_{\tau\supn}\Bigr)\le\Pr(\mathcal{E}).
\end{equation*}
So the remainder of this proof is dedicated to bounding $\Pr(\mathcal{E})$.

Consider the time interval $[0,\tau\supn]$. Let $A$ be the number of task arrivals to $\hat{\bm{W}}^{(k\supn)}$ during this time interval. Then
\begin{align}
\Pr(\mathcal{E})
&=\sum_{j=0}^{\infty}\Pr(A=j)\Pr(\mathcal{E}\mid A=j)\nonumber\\
&\le\sum_{j=0}^{\infty}\frac{(\lambda k\supn\tau\supn
)^je^{-k\supn\lambda\tau\supn}}{j!}j\biggl(1-\frac{\widetilde{\lambda}\supn}{\lambda}\biggr)\label{eq:union}\\
&= k\supn\tau\supn(\lambda-\widetilde{\lambda}\supn),\nonumber
\end{align}
where we have used a union bound for \eqref{eq:union}. By definition,
\begin{align*}
\widetilde{\lambda}\supn&=\frac{\Lambda\supn}{ k\supn}\Biggl(1-\frac{\binom{n-k\supn}{k\supn}}{\binom{n}{k\supn}}\Biggr)\\
%&=\frac{\Lambda\supn}{ k\supn}\Biggl(1-\biggl(1-\frac{k\supn}{n}\biggr)\biggl(1-\frac{k\supn}{n-1}\biggr)\dots\\
%&\mspace{120mu}\cdot\biggl(1-\frac{k\supn}{n-k\supn+1}\biggr)\Biggr)\\
&\ge \frac{\Lambda\supn}{ k\supn}\Biggl(1-\biggl(1-\frac{k\supn}{n}\biggr)^{k\supn}\Biggr)\\
&=\frac{\Lambda\supn}{ k\supn}\biggl(\frac{(k\supn)^2}{n}+O\biggl(\frac{(k\supn)^4}{n^2}\biggr)\biggr)\\
&=\lambda+O\biggl(\frac{(k\supn)^2}{n}\biggr).
\end{align*}
Therefore,
\begin{align*}
\Pr(\mathcal{E})
&= O\Biggl(\biggl(\frac{k\supn}{n^{1/4}}\biggr)^2\Biggr),
\end{align*}
which completes the proof.
\end{proof}

\subsection{Proof of Corollary~\ref{cor:exp}}\label{subsec:proof-response-time-infinite}
\begin{namedthm*}{Corollary~\ref{cor:exp} (Restated)}
Consider an $n$-server system in the limited fork-join model with $k\supn=o(n^{1/4})$, job arrival rate $\Lambda\supn=n\lambda/k\supn$, and exponentially distributed service times with mean $1/\mu$.  Then the steady-state job delay, $T\supn$, converges as:
\begin{equation}\tag*{(\ref{eq:converge-cdf}) (Restated)}
\lim_{n\to\infty}\sup_{\tau\ge 0}\left|\Pr\bigl(T\supn\le \tau\bigr)-\left(1-e^{-(\mu-\lambda)\tau}\right)^{k\supn}\right|=0,
\end{equation}
Specifically, if $k\supn\to\infty$ as $n\to\infty$, then
\begin{equation}\tag*{(\ref{eq:converge-in-distr}) (Restated)}
\frac{T\supn}{H_{k\supn}/(\mu-\lambda)}\tod 1,\quad\text{as }n\to\infty,
\end{equation}
where $H_{k\supn}$ is the $k\supn$-th harmonic number, and further,
\begin{equation}\tag*{(\ref{eq:converge-expectation}) (Restated)}
\lim_{n\to\infty}\frac{\expect\bigl[T\supn\bigr]}{H_{k\supn}/(\mu-\lambda)}=1.
\end{equation}
\end{namedthm*}

\begin{proof}
When the service times are exponentially distributed, each queue is an M/M/1 queue and thus the \cdf\ of the task delay at each queue, $F$, is given by
\begin{equation*}
F(\tau)=1-e^{-(\mu-\lambda)\tau}.
\end{equation*}
Then the convergence in \eqref{eq:converge-cdf} directly follows from Theorem~\ref{thm:asym-independence}.

To prove the weak convergence of $\frac{T\supn}{H_{k\supn}/(\mu-\lambda)}$ in \eqref{eq:converge-in-distr}, we first note that
\begin{equation*}
\frac{\hat{T}\supn}{H_{k\supn}/(\mu-\lambda)}\tod 1,\quad\text{as }n\to\infty,
\end{equation*}
which is a direct implication of the standard result in the asymptotic theory of extremes (see, e.g., Theorem~8.12 in \cite{Das_08}).
Combining this with \eqref{eq:converge-cdf} yields \eqref{eq:converge-in-distr}.

To prove the convergence of the expectation in \eqref{eq:converge-expectation}, we actually need the stochastic dominance shown in Theorem~\ref{thm:independence-upper}. The expectation in \eqref{eq:converge-expectation} can be written as
\begin{align*}
\frac{\expect\bigl[T\supn\bigr]}{H_{k\supn}/(\mu-\lambda)}&=\int_{0}^{\infty}\Pr\biggl(\frac{T\supn}{H_{k\supn}/(\mu-\lambda)}> \tau\biggr)d\tau.
\end{align*}
By Theorem~\ref{thm:independence-upper}, for any $\tau\ge 0$,
\begin{equation*}
\Pr\biggl(\frac{T\supn}{H_{k\supn}/(\mu-\lambda)}> \tau\biggr)\le \Pr\biggl(\frac{\hat{T}\supn}{H_{k\supn}/(\mu-\lambda)}> \tau\biggr).
\end{equation*}
Since
\begin{equation*}
\frac{\hat{T}\supn}{H_{k\supn}/(\mu-\lambda)}\tod 1,\quad\text{as }n\to\infty,
\end{equation*}
and
\begin{align*}
\frac{\expect\bigl[\hat{T}\supn\bigr]}{H_{k\supn}/(\mu-\lambda)}=\int_{0}^{\infty}\Pr\biggl(\frac{\hat{T}\supn}{H_{k\supn}/(\mu-\lambda)}> \tau\biggr)d\tau=1,
\end{align*}
by the General Lebesgue Dominated Convergence Theorem (see, e.g., Theorem~19 in \cite{RoyFit_10}), we can take the limit inside the integral and using \eqref{eq:converge-in-distr}, get
\begin{align*}
\lim_{n\to\infty} \frac{\expect\bigl[T\supn\bigr]}{H_{k\supn}/(\mu-\lambda)}&=\int_{0}^{\infty}\lim_{n\to\infty}\Pr\biggl(\frac{T\supn}{H_{k\supn}/(\mu-\lambda)}> \tau\biggr)d\tau\\
&=\int_{0}^1 1 d\tau\\
&=1,
\end{align*}
which completes the proof.
\end{proof}

\subsection{Explaining $o(n^{1/4})$ and Possible Extensions}\label{sec:explain}
In this section we first explain in a bit more detail where the condition $k\supn=o(n^{1/4})$ comes from.  Recall that in our proof of Theorem~\ref{thm:asym-independence}, we choose a finite time instance $\tau\supn$ and decompose the distance $d_{TV}\left(\pi^{(n,k\supn)},\hat{\pi}^{(k\supn)}\right)$ in Theorem~\ref{thm:asym-independence} into the four distances in Lemmas~\ref{lem:independence-finite-time}--\ref{LEM:PI-TILDE-VS-PI-HAT} accordingly. To understand the result, it helps to return to these lemmas. Instead of directly choosing $\tau\supn$ as $O\Bigl(\frac{n^{1/2}}{k\supn}\Bigr)$, now we keep $\tau\supn$ as a variable.
\begin{itemize}[leftmargin=1em]
\item The distance in \eqref{eq:independence-finite-time} of Lemma~\ref{lem:independence-finite-time} becomes
\begin{equation}\label{eq:independence-taun}
d_{TV}\Bigl(\pi_{\tau\supn}^{(n,k\supn)},\widetilde{\pi}_{\tau\supn}^{(n,k\supn)}\Bigr)=\tau\supn O\left(\frac{(k\supn)^3}{n}\right).
\end{equation}
This is the distance between the limited fork-join system, $\mathcal{S}\supn$, and the system whose first $k\supn$ queues are independent, $\widetilde{\mathcal{S}}\supn$, at time $\tau\supn$. Intuitively, the longer $\tau\supn$ is, the more jobs are expected to arrive during $[0,\tau\supn]$, and thus the more likely it is that the first $k\supn$ queues in $\mathcal{S}\supn$ deviate from the independent queues in $\widetilde{\mathcal{S}}\supn$.  Careful calculation yields that the distance between $\mathcal{S}\supn$ and $\widetilde{\mathcal{S}}\supn$ at time $\tau\supn$ increases linearly with $\tau\supn$ as shown in~\eqref{eq:independence-taun}.

\item The distances in \eqref{eq:convergence} and \eqref{eq:convergence-independent} in Lemma~\ref{lem:convergence-rate} become
\begin{equation}\label{eq:convergence-taun}
d_{TV}\Bigl(\pi_{\tau\supn}^{(n,k\supn)},\pi^{(n,k\supn)}\Bigr)=\frac{1}{\tau\supn}O(k\supn),
\end{equation}
and
\begin{equation}\label{eq:convergence-independent-taun}
d_{TV}\Bigl(\widetilde{\pi}_{\tau\supn}^{(n,k\supn)},\widetilde{\pi}^{(n,k\supn)}\Bigr)=\frac{1}{\tau\supn}O(k\supn).
\end{equation}
These are the distances between the system state at time $\tau\supn$ and the steady state for systems $\mathcal{S}\supn$ and $\widetilde{\mathcal{S}}\supn$, respectively. Intuitively, a long $\tau\supn$ brings the system close to steady state. So as shown in \eqref{eq:convergence-taun} and \eqref{eq:convergence-independent-taun}, these distances decrease with $\tau\supn$.

\item The distance in \eqref{EQ:PI-TILDE-VS-PI-HAT} of Lemma~\ref{LEM:PI-TILDE-VS-PI-HAT} becomes
\begin{equation}\label{EQ:PI-TILDE-VS-PI-HAT-TAUN}
d_{TV}\Bigl(\widetilde{\pi}^{(n,k\supn)},\hat{\pi}^{(k\supn)}\Bigr)=\tau\supn O\left(\frac{(k\supn)^3}{n}\right)+\frac{1}{\tau\supn}O(k\supn).
\end{equation}
This distance has this sum form since it is bounded in a similar way to the distances in Lemma~\ref{lem:independence-finite-time} and \ref{lem:convergence-rate}.
\end{itemize}
To make the sum of the distances in \eqref{eq:independence-taun}--\eqref{EQ:PI-TILDE-VS-PI-HAT-TAUN} as small as possible, we should choose $\tau\supn$ such that these distances are equal, which leads to the choice
\begin{equation*}
\tau\supn=O\left(\frac{n^{1/2}}{k\supn}\right)
\end{equation*}
and a total distance of
\begin{equation*}
O\left(\left(\frac{k\supn}{n^{1/4}}\right)^2\right).
\end{equation*}
Therefore, for this distance to converge to zero, we need the condition that $k\supn=o(n^{1/4})$.

We acknowledge that $k\supn = o(n^{1/4})$ is not the optimal threshold for asymptotic independence to hold. In fact, one can improve the result to $k\supn= o\Bigl(\frac{n^{1/3}}{\log^m n}\Bigr)$ for some $m$ using more delicate bounding techniques in Lemma~\ref{lem:convergence-rate}. 
However, our main contribution is the generalization from the asymptotic independence of a \emph{constant} number of queues to that of a \emph{growing} number of queues.  So we choose to present the case $k\supn = o(n^{1/4})$ to not obscure the main idea with technical details.

As an interesting complement to the asymptotic independence result, we also show that when $k\supn=\Theta(n)$ and the service times are exponentially distributed, any number of queues are \emph{not} asymptotically independent. The formal statement and its proof are given in Theorem~\ref{thm:non-independence} in Appendix~\ref{app:non-independence}.  It then remains an open problem whether there exists a critical value for $k\supn$, where smaller $k\supn$ yields asymptotic independence and larger $k\supn$ does not.

%\begin{figure*}
%\centering
%\begin{subfigure}[t]{0.33\textwidth}
%\centering
%\includegraphics[width=\textwidth]{fig1}
%\caption{$k\supn=n^{1/3}$}
%\label{fig:1-3}
%\end{subfigure}\quad
%\begin{subfigure}[t]{0.33\textwidth}
%\centering
%\includegraphics[width=\textwidth]{fig2}
%\caption{$k\supn=n^{1/2}$}
%\label{fig:1-2}
%\end{subfigure}\\
%\begin{subfigure}[t]{0.33\textwidth}
%\centering
%\includegraphics[width=\textwidth]{fig3}
%\caption{$k\supn=n^{2/3}$}
%\label{fig:2-3}
%\end{subfigure}\quad
%\begin{subfigure}[t]{0.33\textwidth}
%\centering
%\includegraphics[width=\textwidth]{fig4}
%\caption{$k\supn=n^{9/10}$}
%\label{fig:9-10}
%\end{subfigure}
%\caption{Tail distributions of job delays in the limited fork-join systems and the independence upper bounds.}
%\label{fig:ccdf}
%\end{figure*}

\section{Non-Asymptotic Regime: Proof of Independence Upper Bound}\label{sec:independence-upper}
\begin{namedthm*}{Theorem~\ref{thm:independence-upper} (Restated)}
Consider an $n$-server system in the limited fork-join model with $k\supn\le n$. Then the steady-state job delay, $T\supn$, is stochastically upper bounded by the job delay given by independent task delays as defined in \eqref{eq:That}, $\hat{T}\supn$, i.e.,
\begin{equation}\tag*{(\ref{eq:ind-upper}) (Restated)}
T\supn\le_{st} \hat{T}\supn,
\end{equation}
where ``$\le_{st}$'' denotes stochastic dominance. Specifically, for any $\tau\ge 0$,
\begin{align}
\Pr\bigl(T\supn > \tau\bigr)&\le\Pr\bigl(\hat{T}\supn > \tau\bigr) = 1-\left(F(\tau)\right)^{k\supn}.\tag*{(\ref{eq:ind-upper-tail}) (Restated)}
\end{align}
\end{namedthm*}

The main tool we will use is the theory of associated random variables. For convenience of reference, we give the formal definition of association and some properties that we will use in Appendix~\ref{app:association}.  We refer interested readers to \cite{EsaProWal_67} for further details. Intuitively, association is a form of positive correlation among random variables. If a set of random variables are associated, then the maximum of them is stochastically upper bounded by the maximum of independent versions of them. Since the delay of a job is the maximum of its task delays, to show Theorem~\ref{thm:independence-upper}, it thus suffices to show association of the task delays. This further boils down to showing association among the workloads of any $k\supn$ queues in steady state since each task delay is the workload of the queue that the task is sent to plus the service time of the task.

Such an association result has been proven for the classical fork-join model \cite{NelTan_88} where $k\supn = n$, but the approach of the proof breaks down once we have $k\supn < n$.  The proof idea there is to observe the system at each job arrival time, and show that the numbers of tasks sent to different queues are associated.  This proof idea is widely used in the literature to establish association (see, e.g., \cite{KumSho_93,ShaBouBac_17}), but it does not work in the limited fork-join model when $k\supn<n$.  We can think of the process of assigning a job's tasks to queues as a balls-and-bins problem, where the $k\supn$ tasks correspond to $k\supn$ balls, the queues are the bins, and the number of balls thrown in each bin is the number of tasks sent to each queue.  When $k\supn = n$, it is obvious that the numbers of balls in the bins are associated since they are all exactly equal to one.  But when $k\supn < n$, the numbers of balls in the bins are actually \emph{negatively associated} by a classical result \cite{JoaPro_83}! However, one should not be discouraged since this does not mean that the \emph{steady-state} workloads are negatively associated.

In our proof, we develop a novel technique that we call ``Poisson oversampling'', where we observe the system not only when jobs arrive but also at the jump times of a Poisson process that is independent of everything else.  In the existing approach where the system is observed only at job arrival times, there is always one job arrival at each observation time.  But with oversampling, there could be one or zero job arrivals at each observation time.  Recall that jobs arrive with rate $\Lambda\supn$.  Let the additional Poisson process have rate $\beta\supn$.  Then in the corresponding balls-and-bins problem, with probability $\Lambda\supn/(\Lambda\supn+\beta\supn)$, $k\supn$ balls are thrown into $k\supn$ distinct bins chosen uniformly at random, and with probability $\beta\supn/(\Lambda\supn+\beta\supn)$, there are no balls at all.  We will see in the proof that, surprisingly, now the numbers of balls thrown into any $k\supn$ bins become associated with properly chosen $\beta\supn$.  This enables us to show association of steady-state workloads.

\begin{remark*}
Before we present the proof, we remark that it may be possible to explore the monotonicity of the workload process to establish association \cite{Har_77,Cox_84,Lig_05}.\footnote{We thank Prof.\ XYZ for suggesting this possible approach.}  However, the results in \cite{Har_77,Cox_84,Lig_05} assume either a finite or a compact state space.  It may be possible to generalize the results from a finite state space \cite{Har_77,Cox_84} to a countable state space for some Markov chains, which then can be applied to our problem for certain phase-type service time distributions.  To further deal with more general service time distributions, we may be able to utilize the existing results for a compact state space \cite{Lig_05}.  But there we need to compactify the state space and verify a condition on the generator of the workload process.  We do not pursue such an approach here.
\end{remark*}

\begin{proof}
\begin{sloppypar}
Recall that
\begin{equation*}
T\supn=\max\bigl\{T\supn_1,T\supn_2,\dots,T\supn_{k\supn}\bigr\},
\end{equation*}
where $T\supn_i$ denotes the steady-state task delay at queue $i$.  Then to prove the stochastic dominance, it suffices to prove that $T\supn_1$, $T\supn_2,\dots,T\supn_{k\supn}$ are associated \cite[Theorem~5.1]{EsaProWal_67}.
%; i.e., for any \emph{nondecreasing} functions $f$ and $g$ of $\bm{T}=\bigl(T\supn_1,T\supn_2,\dots,T\supn_{k\supn}\big)$,
%\begin{equation}\label{eq:association-T}
%\expect[f(\bm{T})g(\bm{T})]\ge \expect[f(\bm{T})]\expect[g(\bm{T})].
%\end{equation}
\end{sloppypar}

\begin{sloppypar}
We start by noting that it is sufficient to prove that the steady-state workloads, $W\supn_1(\infty),\allowbreak W\supn_2(\infty),\dots,W\supn_{k\supn}(\infty)$, are associated.  The sufficiency follows from the fact that each $T\supn_i$ can be expressed in the following form:
\begin{equation*}
T\supn_i=W\supn_i(\infty)+X_i,
\end{equation*}
where $X_1,\allowbreak X_2,\dots,X_{k\supn}$ represent the service times of tasks so they are i.i.d.$\sim G$ and independent of everything else.  Then $T\supn_1,\allowbreak T\supn_2,\dots,T\supn_{k\supn}$ are nondecreasing functions of $W\supn_1(\infty),\allowbreak W\supn_2(\infty),\dots,W\supn_{k\supn}(\infty),\allowbreak  X_1,\allowbreak X_2,\dots,X_{k\supn}$. So by Properties (P1) and (P2) in Lemma~\ref{lem:association}, $T\supn_1,\allowbreak T\supn_2,\dots,T\supn_{k\supn}$ are associated when $W\supn_1(\infty),\allowbreak W\supn_2(\infty),\dots,W\supn_{k\supn}(\infty)$ are associated. 
\end{sloppypar}

\begin{sloppypar}
All that remains is prove the claim that $W\supn_1(\infty),\allowbreak W\supn_2(\infty),\dots,\allowbreak W\supn_{k\supn}(\infty)$ are associated.  We will work with a discrete-time Markov chain constructed from the continuous-time workload process $(\bm{W}\supn(t),t\ge 0)$.  Specifically, we consider a Poisson process, denoted by $(B(t),t\ge 0)$, that is independent of everything else.  Let the rate of this Poisson process be $\beta\supn$, which will be specified later in \eqref{eq:beta}.  Then we sample the workload process $(\bm{W}\supn(t),t\ge 0)$ at time instances right before either a job arrival or an event of the Poisson process $(B(t),t\ge 0)$.  Let such time instances be denoted by $\{U_s,s=0,1,\dots\}$ with $U_0=0$.  This gives us a discrete-time Markov chain, which we denote by $(\bm{\Phi}\supn(s),s=0,1,\dots)$, i.e., $\Phi\supn_i(s)=W\supn_i(U_s^-)$, where $W\supn_i(U_s^-)$ is the workload of queue $i$ right before time $U_s$.  Since $(\bm{\Phi}\supn(s),s=0,1,\dots)$ is constructed by sampling the workload process more often than the job arrival process, we call this technique ``Poisson oversampling''.

We first claim that $(\bm{\Phi}\supn(s),s=0,1,\dots)$ converges to a well-defined steady state $\bm{\Phi}\supn(\infty)$ and that $\bm{\Phi}\supn(\infty)$ and $\bm{W}\supn(\infty)$ are identically distributed.  This claim can be proven by showing that $(\bm{\Phi}\supn(s),s=0,1,\dots)$ is aperiodic and positive Harris recurrent and then appealing to the PASTA property \cite{MelWhi_90}.  We omit the proof of this claim since the aperiodicity is straightforward to check and the positive Harris recurrence follows from the rather standard Foster-Lyapunov criteria using the quadratic Lyapunov function \cite{MeyTwe_92}.  With this claim, it then suffices to prove that $\Phi\supn_1(\infty),\allowbreak \Phi\supn_2(\infty),\dots,\allowbreak\Phi\supn_{k\supn}(\infty)$ are associated.

We assume that $\Phi\supn_i(0)=0$ for every $i=1,2,\dots,n$.  We will prove that $\Phi\supn_1(s),\Phi\supn_2(s),\dots,\Phi\supn_{k\supn}(s)$ are associated for any $s\ge 0$ by induction on $s$. Then $\Phi\supn_1(\infty),\allowbreak\Phi\supn_2(\infty),\dots,\allowbreak\Phi\supn_{k\supn}(\infty)$ are associated since $\bm{\Phi}\supn(s)\tod \bm{\Phi}\supn(\infty)$ as $s\to\infty$ \cite{EsaProWal_67}.
\end{sloppypar}

\emph{Base Step:} $\Phi\supn_1(0),\Phi\supn_2(0),\dots,\Phi\supn_{k\supn}(0)$ are associated since they are all zero.
\begin{figure*}
\centering
\includegraphics[width=0.3\textwidth]{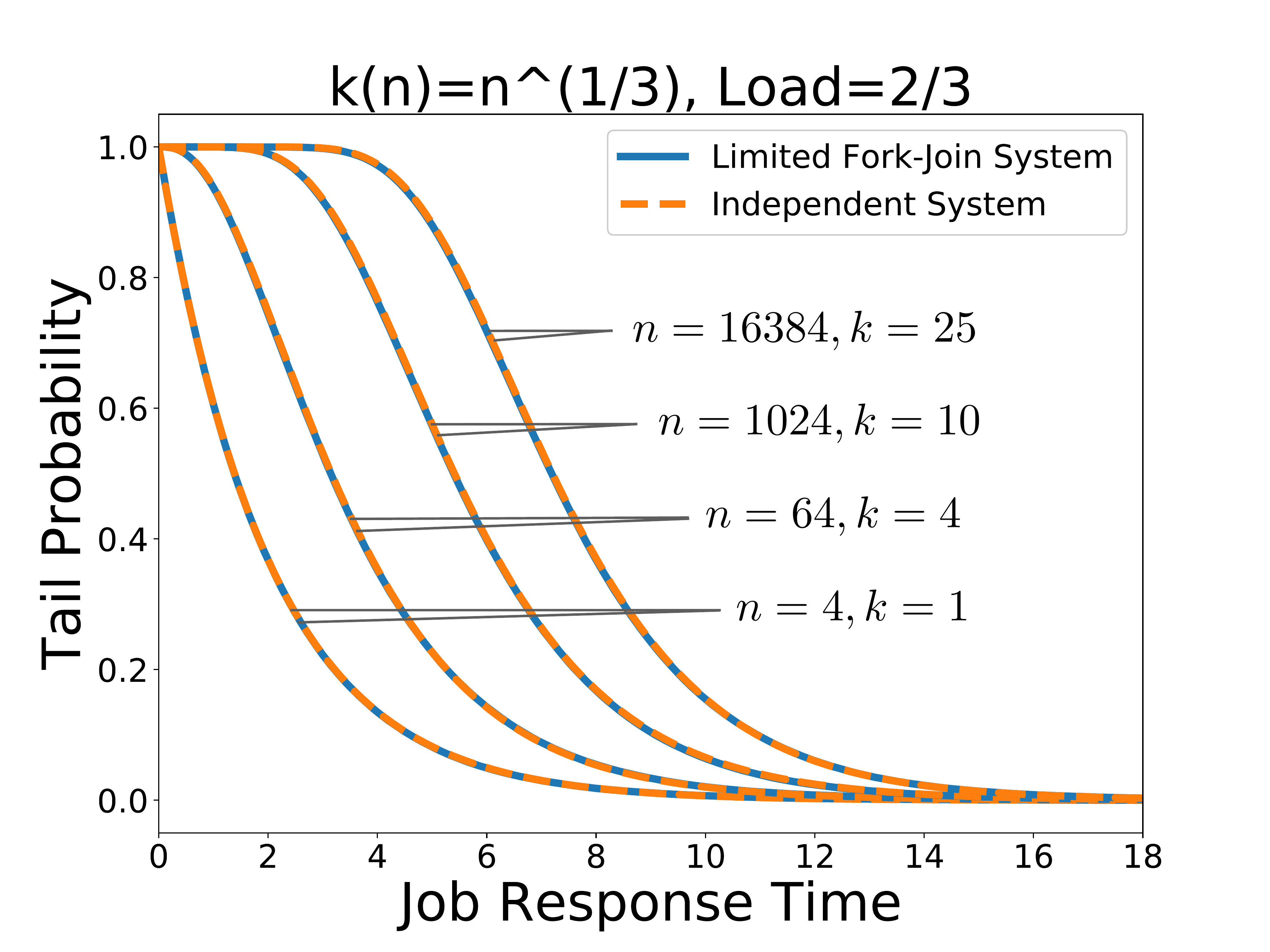}
\quad
\includegraphics[width=0.3\textwidth]{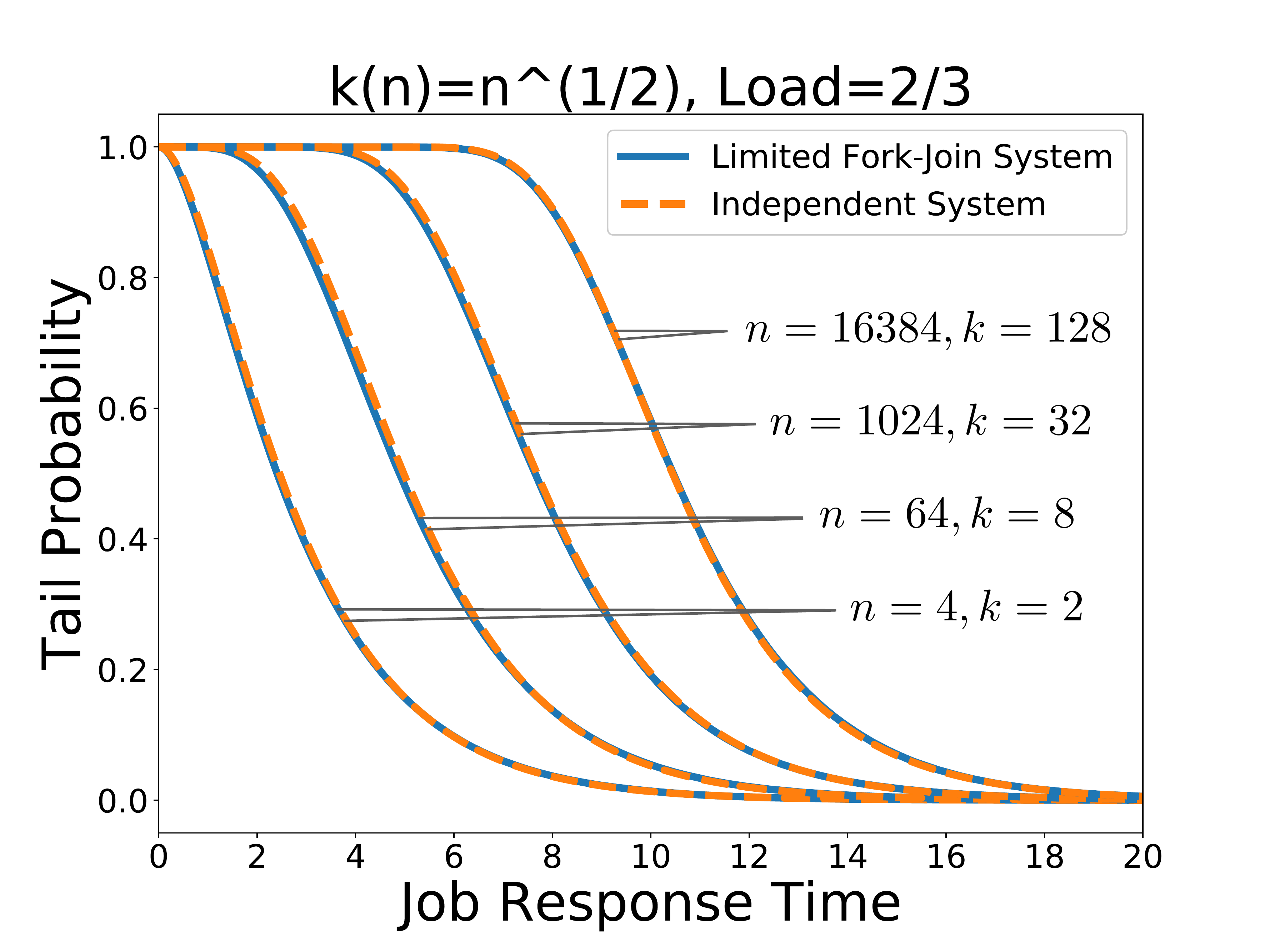}
\\
\includegraphics[width=0.3\textwidth]{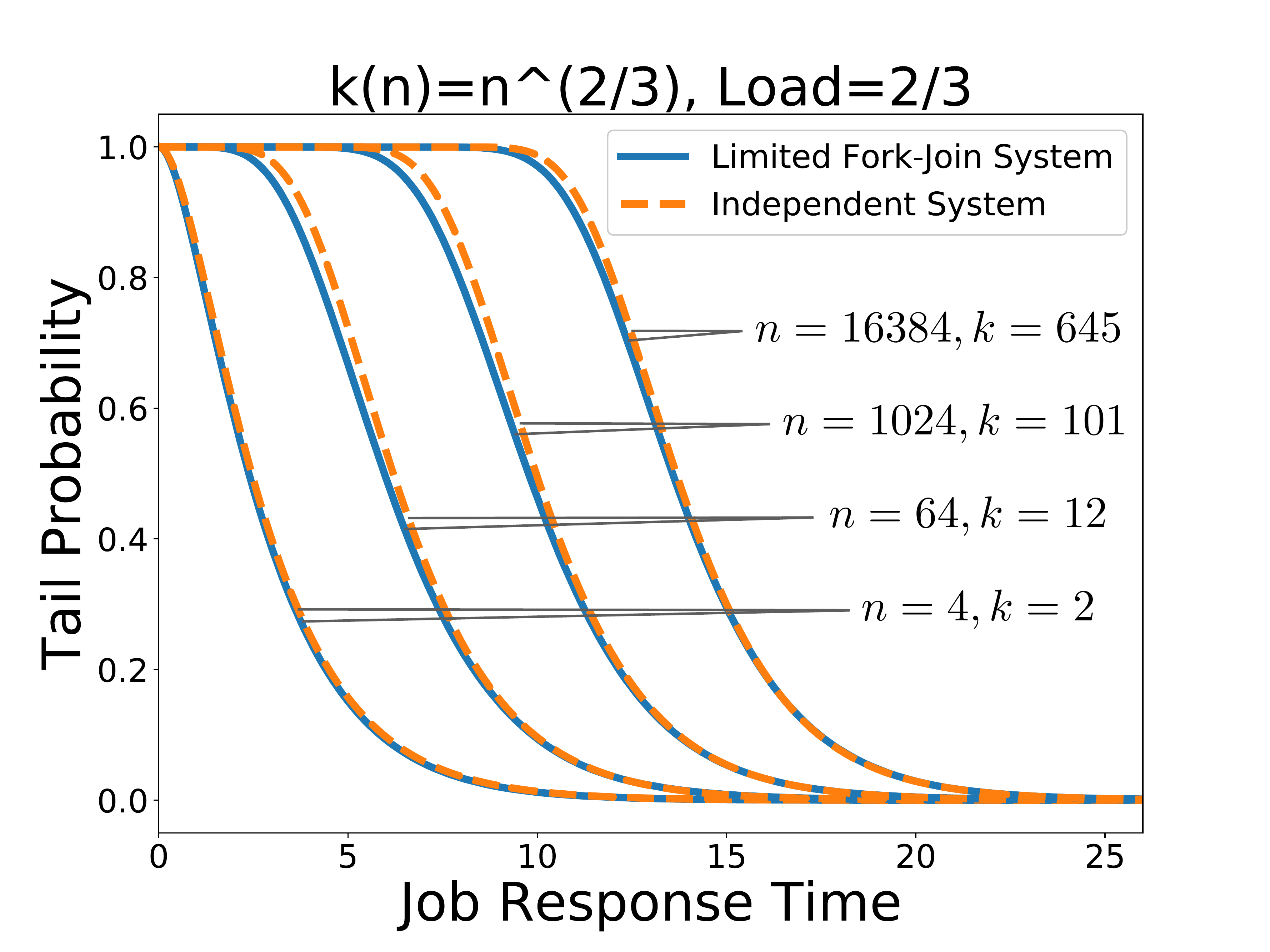}
\quad
\includegraphics[width=0.3\textwidth]{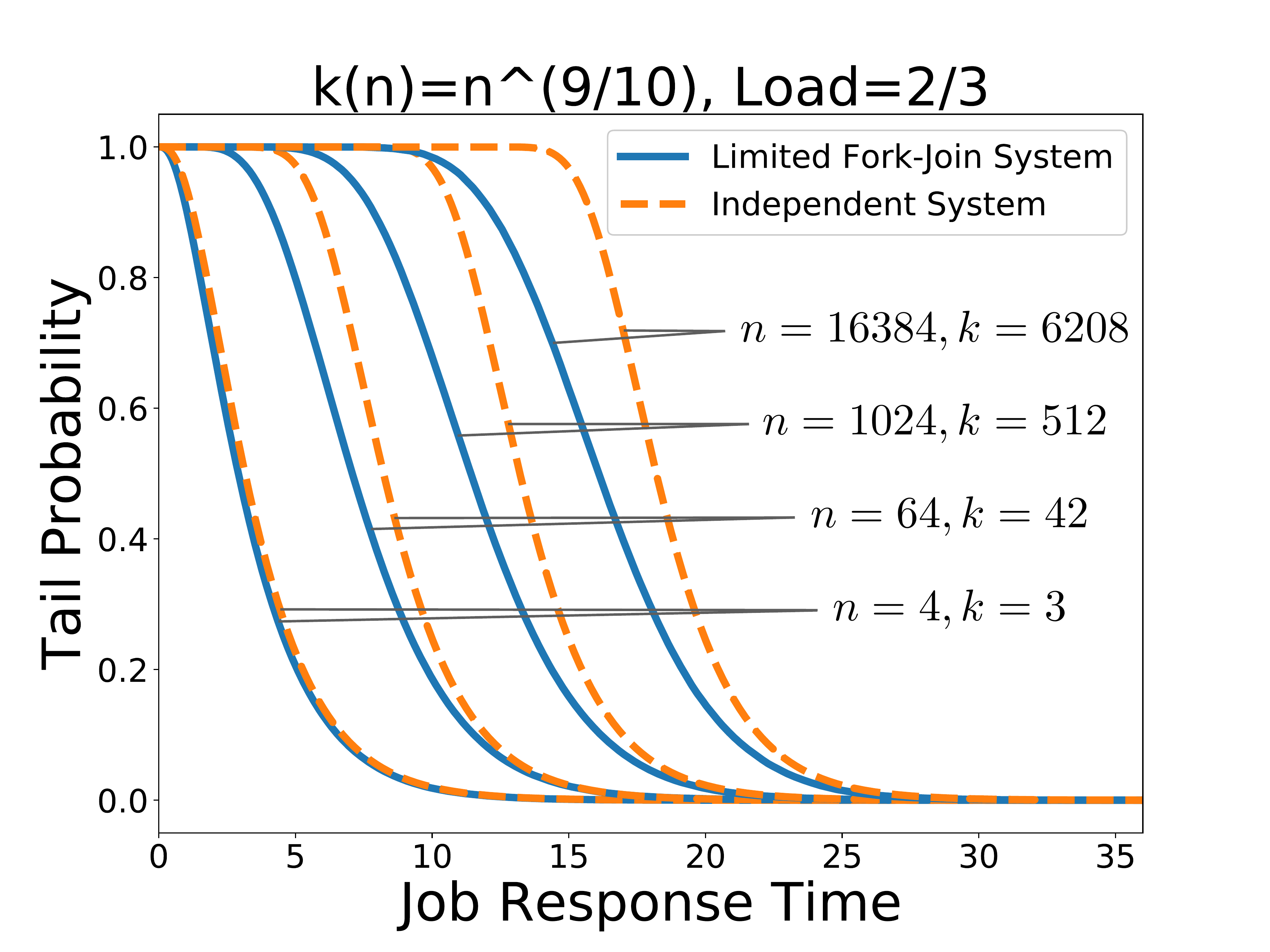}
\caption{Tail distributions of job delays in the limited fork-join systems and the independence upper bounds.}
\label{fig:ccdf}
\end{figure*}

\begin{sloppypar}
\emph{Inductive Step:}
Assuming that $\Phi\supn_1(s),\Phi\supn_2(s),\dots,\Phi\supn_{k\supn}(s)$ are associated for some $s\ge 0$, we will show that $\Phi\supn_1(s+1),\allowbreak\Phi\supn_2(s+1),\dots,\allowbreak\Phi\supn_{k\supn}(s+1)$ are associated. By Lindley equation,
\begin{equation}\label{eq:lindley}
\Phi\supn_i(s+1)=\Bigl(\Phi\supn_i(s)+ Y_i(s) - \Delta U(s)\Bigr)^+,
\end{equation}
where $Y_i(s)$ is the service time needed by the task that arrives to queue $i$ at time $U_s$, and $\Delta U(s)=U_{s+1}-U_s$.  Note that at time $U_s$, there may be no task arrival to queue $i$, either because there is no job arrival or because there is a job arrival but it does not send any tasks to queue $i$. So we can write $Y_i(s)$ as
\begin{equation*}
Y_i(s)=A_i(s)\cdot X_i(s),
\end{equation*}
where $A_i(s)$ equals to either $1$ or $0$, representing the number of task arrivals to queue $i$ at time $U_{s}$, and $X_i(s)$ is a r.v.\ with distribution $G$ and is independent of everything else, representing the service time.  Then $\Phi\supn_i(s+1),i=1,2,\dots,k\supn$ are nondecreasing functions of the $\Phi\supn_i(s)$'s, $A_i(s)$'s, $X_i(s)$'s and $-\Delta U(s)$ with $i=1,2,\dots,k\supn$.  We can see that each of the following four sets of r.v.'s, $\bigl\{\Phi\supn_1(s),\allowbreak \Phi\supn_2(s),\dots,\allowbreak \Phi\supn_{k\supn}(s)\bigr\}$, $\{A_1(s),\allowbreak A_2(s),\dots,\allowbreak A_{k\supn}(s)\}$, $\{X_1(s),\allowbreak X_2(s),\dots,\allowbreak X_{k\supn}(s)\}$, and $\{-\Delta U(s)\}$, is independent of the union of others.  So to show that $\Phi\supn_i(s+1),i=1,2,\dots,k\supn$ are associated, it suffices to show that each of these sets is a set of associated r.v.'s.
\end{sloppypar}

(i) The $\Phi\supn_i(s),i=1,2,\dots,k\supn$ are associated by assumption.

(ii) The $X_i(s),i=1,2,\dots,k\supn$ are associated since they are independent.

(iii) The r.v.\ $-\Delta U(s)$ is associated since a single r.v.\ is associated.

(iv) We now prove that $A_i(s),i=1,2,\dots,k\supn$ are associated. We note that here they do not satisfy the lattice condition in the celebrated FKG inequality \cite{ForKasGin_71}. For conciseness of notation, let $\bm{A}=(A_1(s),A_2(s),\dots,A_{k\supn}(s))$. To show association, it suffices to prove that for all \emph{binary-valued}, \emph{(entrywisely) nondecreasing} functions $f$ and $g$ \cite{EsaProWal_67},
\begin{equation}\label{eq:association}
\expect[f(\bm{A})g(\bm{A})]\ge \expect[f(\bm{A})]\expect[g(\bm{A})].
\end{equation}
By construction, it is clear that $\bm{A}\in\{0,1\}^{k\supn}$.  If either $f$ or $g$ always has constant value $0$ or $1$, then \eqref{eq:association} trivially holds.  So we can focus on the case that neither $f$ nor $g$ is a constant function.  In this case, by the monotonicity of $f$ and $g$, we have $f\bigl((0,\dots,0)\bigr)=g\bigl((0,\dots,0)\bigr)=0$ and $f\bigl((1,\dots,1)\bigr)=g\bigl((1,\dots,1)\bigr)=1$. Note that at each sample time $U_s$, the probability that there is a job arrival is $\Lambda\supn/(\Lambda\supn+\beta\supn)$. Then
\begin{align*}
\expect[f(\bm{A})g(\bm{A})]&\ge \Pr\bigl(\bm{A}=(1,\dots,1)\bigr)\cdot f\bigl((1,\dots,1)\bigr)g\bigl((1,\dots,1)\bigr)\\
&=\frac{\Lambda\supn}{\Lambda\supn+\beta\supn} \frac{1}{\binom{n}{k\supn}}.
\end{align*}
Since $f(\bm{a})\le 1,g(\bm{a})\le 1$ for any $\bm{a}\in\{0,1\}^{k\supn}$,
\begin{align*}
\expect[f(\bm{A})]&= \sum_{\substack{\bm{a}\in\{0,1\}^{k\supn}:\\\bm{a}\neq (0,\dots,0)}}\Pr(\bm{A}=\bm{a})f(\bm{a})\\
&\le \Pr(\bm{A}\neq (0,\dots,0))\\
&=\frac{\Lambda\supn}{\Lambda\supn+\beta\supn}p,
\end{align*}
where $p$ is the probability that a job arrival does not sent tasks to queues $1,2,\dots,k\supn$, so $p$ does not depend on $\beta\supn$ and
\begin{equation*}
p=
\begin{cases}
1 & \text{ if }k\supn>n/2,\\
\frac{1}{\binom{n}{k\supn}}\left(\binom{n}{k\supn}-\binom{n-k\supn}{k\supn}\right) & \text{ if }k\supn \le n/2.
\end{cases}
\end{equation*}
Similarly,
\begin{align*}
\expect[g(\bm{A})]
&\le\frac{\Lambda\supn}{\Lambda\supn+\beta\supn}p.
\end{align*}
We choose any $\beta\supn$ such that
\begin{equation*}
\frac{\Lambda\supn}{\Lambda\supn+\beta\supn} \frac{1}{\binom{n}{k\supn}}\ge \left(\frac{\Lambda\supn}{\Lambda\supn+\beta\supn}\right)^2p^2,
\end{equation*}
i.e., any $\beta\supn$ such that
\begin{align}
\beta\supn\ge\Lambda\supn\left(\binom{n}{k\supn}p^2-1\right).\label{eq:beta}
\end{align}
Then
\begin{equation*}
\expect[f(\bm{A})g(\bm{A})]\ge \expect[f(\bm{A})]\expect[g(\bm{A})],
\end{equation*}
which completes the induction, and thus completes the proof.
\end{proof}

\section{Evaluation via Simulations}\label{sec:simulations}
In this section we use simulation to explore the regimes of $k\supn$ that are not covered by our theoretical analysis.  Specifically, our theoretical analysis has established that when $k\supn=o(n^{1/4})$, any $k\supn$ queues are asymptotically independent and the job delay converges to the independence upper bound; when $k\supn=\Theta(n)$, any number of multiple queues are bounded away from being independent.  We therefore simulate the limited fork-join systems for the following four settings between $o(n^{1/4})$ and $\Theta(n)$: $k\supn=n^{1/3}$, $k\supn=n^{1/2}$, $k\supn=n^{2/3}$ and $k\supn=n^{9/10}$.  We simulate the $n$-server system for $n=4, 64, 1024$ and $16384$ under each setting.

We compare the tail distribution (complementary cumulative distribution function) of the job delay in each limited fork-join system with the independence upper bound. Figure~\ref{fig:ccdf} shows the results for systems with exponentially distributed service times and load $\rho=2/3$ on each individual queue.  We see that for $k\supn=n^{1/3}$, the independence upper bound is strikingly accurate.  For $k\supn=n^{1/2}$, the gap between the job delay and the independence upper bound seems to be diminishing when $n$ is large enough.  But for $k\supn=n^{2/3}$, it is rather unclear if the job delay will converge to the independence upper bound or not.  Finally, when $k\supn=n^{9/10}$, the job delay evidently \emph{diverges} from the independence upper bound.  We have also simulated systems for different loads ($\rho = 1/3, 0.9$) and different service time distributions (deterministic, truncated Pareto, hyperexponential), and similar phenomena are observed.

\section{Conclusions}\label{sec:conclusions}
We study the limited fork-join model where there are $n$ servers in the system and each job consists of $k\supn\le n$ tasks that are sent to $k\supn$ distinct servers chosen uniformly at random.  A job is considered complete only when all its tasks complete processing.  We characterize the delay of jobs both in an asymptotic regime where $n\to\infty$ and in the non-asymptotic regime for any $n$ and any $k\supn=k$.

For the asymptotic regime, we show that under the condition $k\supn=o(n^{1/4})$, the workloads of any $k\supn$ queues in the $n$-server system are asymptotically independent, and the delay of a job therefore converges to the maximum of independent task delays.  For the non-asymptotic regime, we show that the steady-state workloads of any $k\supn$ queues are associated, and therefore assuming independent task delays yields an upper bound on the job delay.  Our results provide the first tight characterization of job delay in the limited fork-join model, and the upper bound is tighter than other existing upper bounds.

From a technical perspective, we make the following two contributions: (1) Our asymptotic results open up new regimes for asymptotic independence: $k\supn$ queues are shown to be asymptotically independent, where $k\supn$ is allowed to grow with $n$ instead of being a constant, as was previously studied. (2) We develop new proof techniques to establish association in steady state. We believe that the results and techniques in this paper will shed light on related problems such as order statistics in coded data storage systems, job redundancy, load-balancing algorithms.

\section{Acknowledgment}
\begin{sloppypar}
This work was supported in part by NSF Grants CPS ECCS-1739189, ECCS 1609370, XPS-1629444, and CMMI-1538204, the U.S.\ Army Research Office (ARO Grant No.\ W911NF-16-1-0259), the U.S.\ Office of Naval Research (ONR Grant No.\ N00014-15-1-2169), DTRA under the grant number HDTRA1-16-0017, and a 2018 Faculty Award from Microsoft. Additionally, Haotian Jiang was supported in part by the Department of Physics at Tsinghua University.
\end{sloppypar}

%\bibliographystyle{ACM-Reference-Format}
%%% -*-BibTeX-*-
%%% Do NOT edit. File created by BibTeX with style
%%% ACM-Reference-Format-Journals [18-Jan-2012].

%\bibliography{/Users/weinawang/bibiCloud/refs-weina}
\citestyle{acmnumeric}

\appendix
\section{Non-independence Result for $k\supn=\Theta(n)$}\label{app:non-independence}
\begin{theorem}\label{thm:non-independence}
Consider an $n$-server system in the limited fork-join model with $k\supn=\Theta(n)$, job arrival rate $\Lambda\supn=n\lambda/k\supn$, and exponentially distributed service times with rate $\mu$. Let $\pi^{(n,2)}$ denote the joint distribution of the steady-state queue lengths for any two queues in the $n$-server system.  Let $\hat{\pi}^{(2)}$ denote the joint distribution of the steady-state queue lengths of two independent M/M/1 queues, each with load $\rho$.  Then there exists an $\epsilon>0$ and $n_0>0$, such that for any $n>n_0$, $d_{TV}\bigl(\pi^{(n,2)}, \hat{\pi}^{(2)}\bigr)>\epsilon$.
\end{theorem}
\begin{proof}
We assume that $k\supn = p n$ for a constant $p$ with $0<p\le 1$.  Then the job arrival rate is given by $\Lambda\supn = \lambda/p$, which is a constant.  So we rewrite $\Lambda\supn$ as $\Lambda$ for conciseness.

Let $\epsilon=\frac{p\lambda(1-\rho)^2}{2(11\Lambda+8\mu)}$.  We will specify $n_0$ later.  Suppose by contradiction that $d_{TV}\bigl(\pi^{(n,2)},\hat{\pi}^{(2)}\bigr)\le \epsilon$ for all $n>n_0$.  We will show that this assumption contradicts with the balance equations of the first two queues in the limited fork-join system with $n$ servers.

We first write out the balance equations for the Markov chain formed by the queue lengths of the first two queues. Consider a job arrival to this $n$-server system.  Let $p_0\supn$ be the probability that no task arrives to the first two queues, and $p_1\supn$ be the probability that exactly one task arrives to the first two queues. Let $p_2\supn=1-p_0\supn-p_1\supn$ be the probability that two tasks arrive to the first two queues.  We can compute these probabilities as follows:
\begin{align*}
p_0\supn &= \binom{n-2}{k}\Bigm / \binom{n}{k} \rightarrow p_0:=(1-p)^2 \quad \text{as }n\to \infty,\\
p_1\supn &= \frac{2\binom{n-2}{k-1}}{\binom{n}{k}} \rightarrow p_1:=2p(1-p) \quad \text{as }n \rightarrow \infty,\\
p_2\supn &= \frac{\binom{n-2}{k-2}}{\binom{n}{k}} \rightarrow p_2:=p^2 \quad \text{as }n \rightarrow \infty.
\end{align*}
Recall that the joint distribution of the steady-state queue lengths of the first two queues is $\pi^{(n,2)}$.  Then the balance equation of the first two queues for the state $(1,1)$ can be written as
\begin{align}
0&=\pi^{(n,2)}(1,1) \cdot (p_1\supn\Lambda+p_2\supn\Lambda+2\mu)\nonumber\\
&\mspace{23mu}-\biggl(\frac{1}{2}\pi^{(n,2)}(0,1)p_1\supn\Lambda+\frac{1}{2}\pi^{(n,2)}(1,0)p_1\supn\Lambda\label{eqn: balance}\\
&\mspace{57mu}+\pi^{(n,2)}(0,0)p_2\supn\Lambda+\pi^{(n,2)}(1,2)\mu+\pi^{(n,2)}(2,1)\mu\biggr).\nonumber
\end{align}

Let the right-hand-side of \eqref{eqn: balance} be denoted by $\mathcal{R}(\pi^{(n,2)})$. Let
\begin{align*}
a_1&=(p_1\supn-p_1)\Lambda\biggl(\pi^{(n,2)}(1,1)-\frac{1}{2}\pi^{(n,2)}(0,1)-\frac{1}{2}\pi^{(n,2)}(1,0)\biggr)\\
&\mspace{23mu}+(p_2\supn-p_2)\Lambda\biggl(\pi^{(n,2)}(1,1)-\frac{1}{2}\pi^{(n,2)}(0,1)\\
&\mspace{150mu}-\frac{1}{2}\pi^{(n,2)}(1,0)-\pi^{(n,2)}(0,0)\biggr),\\
a_2&=(\pi^{(n,2)}(1,1)-\hat{\pi}^{(2)}(1,1))(p_1\Lambda+p_2\Lambda+2\mu)\\
&\mspace{23mu}-\frac{1}{2}(\pi^{(n,2)}(0,1)-\hat{\pi}^{(2)}(0,1)+\pi^{(n,2)}(1,0)-\hat{\pi}^{(2)}(1,0))p_1\Lambda\\
&\mspace{23mu}-(\pi^{(n,2)}(0,0)-\hat{\pi}^{(2)}(0,0))p_2\Lambda\\
&\mspace{23mu}-(\pi^{(n,2)}(1,2)-\hat{\pi}^{(2)}(1,2)+\pi^{(n,2)}(2,1)-\hat{\pi}^{(2)}(2,1))\mu.
\end{align*}
Then since $\hat{\pi}^{(2)}(q_1,q_2)=(1-\rho)^2\rho^{q_1+q_2}$ for any $(q_1,q_2)\in\mathbb{Z}_+^2$,
\begin{align}
\mathcal{R}(\pi^{(n,2)})&=a_1+a_2+\hat{\pi}^{(2)}(1,1) \cdot (p_1\Lambda+p_2\Lambda+2\mu)\nonumber\\
&\mspace{23mu}-\biggl(\frac{1}{2}\hat{\pi}^{(2)}(0,1)p_1\Lambda+\frac{1}{2}\hat{\pi}^{(2)}(1,0)p_1\Lambda\nonumber\\
&\mspace{57mu}+\hat{\pi}^{(2)}(0,0)p_2\Lambda+\hat{\pi}^{(2)}(1,2)\mu+\hat{\pi}^{(2)}(2,1)\mu\biggr)\nonumber\\
&=a_1+a_2-p \lambda(1-\rho)^4.\label{eq:balance-ind}
\end{align}

We choose $n_0$ such that for any $n> n_0$, $|p_1\supn-p_1|\le \epsilon$ and $|p_2\supn-p_2|\le \epsilon$.  Then it is not hard to see that $|a_1|\le 3\Lambda\epsilon$.  By the assumption that $d_{TV}\bigl(\pi^{(n,2)},\hat{\pi}^{(2)}\bigr)\le \epsilon$, we have that $|a_2|\le8(\Lambda+\mu) \epsilon$. By the choice of $\epsilon$, $|a_1+a_2|\le (11\Lambda+8\mu)\epsilon=\frac{1}{2}p\lambda(1-\rho)^2$.  Therefore, $\mathcal{R}(\pi^{(n,2)})<0$ by \eqref{eq:balance-ind}, which contradicts with the balance equation \eqref{eqn: balance}.  This completes the proof of Theorem~\ref{thm:non-independence}.
\end{proof}

\section{Definition and Some Properties of Association}\label{app:association}
\begin{definition}[Association \cite{EsaProWal_67}]
We say random variables $X_1$, $X_2,\dots,X_m$ are associated if for all (entrywisely) nondecreasing functions $f$ and $g$,
\begin{multline}
\expect[f(X_1,X_2,\dots,X_m)g(X_1,X_2,\dots,X_m)]\\
\ge \expect[f(X_1,X_2,\dots,X_m)]\expect[g(X_1,X_2,\dots,X_m)].
\end{multline}
\end{definition}

\begin{lemma}[\cite{EsaProWal_67}]\label{lem:association}
Associated random variables have the following properties:
\begin{enumerate}[label=(P\arabic*),leftmargin=3em]
\item Nondecreasing functions of associated random variables are associated.
\item If two sets of associated random variables are independent of one another, then their union is a set of associated random variables.
\item If a sequence of random vectors $\bm{X}(u)\tod \bm{X}$ as $u\to\infty$ and for each $u$, the entries of $\bm{X}(u)$ are associated, then the entries of $\bm{X}$ are associated.
\end{enumerate}
\end{lemma}

\end{document}